\documentclass[preprint,11pt,times,authoryear]{elsarticle}

\usepackage[section]{placeins}
\usepackage{mathdots}
\usepackage{dsfont}
\usepackage{amssymb}
\usepackage{amsthm}
\newtheorem{theorem}{Theorem}

\newtheorem{proposition}[theorem]{Proposition}

\newtheorem{remark}[theorem]{Remark}
\newtheorem{definition}{Definition}
\newtheorem{example}{Example}

\usepackage{amsmath}
\usepackage{xspace}
\usepackage[algo2e,boxed,vlined,linesnumberedhidden,
titlenumbered]{algorithm2e}
\SetKw{KwFrom}{from}
\SetKwFor{For}{For}{do}{End for}
\SetKwFor{Forall}{For all}{do}{End for}
\SetKw{KwRet}{Return}
\SetKw{KwErr}{Error}
\SetKw{KwAnd}{and}
\SetKw{KwOr}{or}
\SetKw{KwBreak}{break}
\SetKw{KwContinue}{continue}
\SetKw{KwTrue}{true}
\SetKw{KwFalse}{false}
\SetKwFor{While}{While}{do}{End do}
\SetKwIF{If}{ElseIf}{Else}{If}{then}{Else if}{Else}{End if}
\usepackage{kbordermatrix}
\usepackage{pgfplots}
\pgfplotsset{compat=newest,compat/show suggested version=false}
\pgfplotsset{
    legend image with text/.style={
        legend image code/.code={%
            \node[anchor=center] at (0.3cm,0cm) {#1};
        }
    },
}

\usepackage{ytableau}
\usepackage{url,hyperref}
\renewcommand{\kbldelim}{(}
\renewcommand{\kbrdelim}{)}

\newcommand\mcal{\mathcal}
\newcommand\mbf{\mathbf}
\newcommand\mbb{\mathbb}

\newcommand\x{\mbf{x}}

\newcommand\bin{\mbf{b}}
\newcommand\bi{\mbf{i}}
\newcommand\bj{\mbf{j}}
\newcommand\bk{\mbf{k}}
\newcommand\bd{\mbf{d}}

\newcommand\bu{\mbf{u}}

\newcommand\balpha{\boldsymbol\alpha}

\newcommand\N{\mbb{N}}

\newcommand\K{\mbb{K}}

\newcommand\R{\mbb{R}}

\newcommand\cG{\mcal{G}}
\newcommand\cK{\mcal{K}}

\newcommand\cS{\mcal{S}}
\newcommand\cT{\mcal{T}}

\newcommand\ceil[1]{\left\lceil#1\right\rceil}
\newcommand\floor[1]{\left\lfloor#1\right\rfloor}
\newcommand\cro[1]{\left[#1\right]}
\newcommand\pare[1]{\left(#1\right)}
\newcommand\acc[1]{\left\{#1\right\}}

\newcommand\BM{\textsc{BM}\xspace}
\newcommand\BMS{\textsc{BMS}\xspace}

\newcommand\FGLM{\textsc{FGLM}\xspace}
\newcommand\sFGLM{\textsc{Scalar-FGLM}\xspace}
\newcommand\asFGLM{\textsc{Adaptive Scalar-FGLM}\xspace}

\newcommand\spFGLM{\textsc{Sparse-FGLM}\xspace}
\newcommand\bms{Berlekamp~-- Massey~-- Sakata\xspace}
\newcommand\bm{Berlekamp~-- Massey\xspace}
\newcommand\gb{Gr\"obner basis\xspace}
\newcommand\gbs{Gr\"obner bases\xspace}
\newcommand\ie{\mbox{{i.e.}}\xspace}
\newcommand\resp{\mbox{resp.}\xspace}

\DeclareMathOperator\DRL{\textsc{drl}}
\DeclareMathOperator\LEX{\textsc{lex}}

\newcommand{\adots}{\mathinner{%
  \mkern1mu\raise1pt\hbox{.}%
  \mkern2mu\raise4pt\hbox{.}%
  \mkern2mu\raise7pt\vbox{\kern7pt\hbox{.}}\mkern1mu}} 
\DeclareMathOperator\rank{rank}
\DeclareMathOperator\LM{\textsc{lm}}
\DeclareMathOperator\LT{\textsc{lt}}
\DeclareMathOperator\LC{\textsc{lc}}
\DeclareMathOperator\lcm{lcm}
\DeclareMathOperator\Stabilize{Stabilize}
\DeclareMathOperator\Border{Border}
\DeclareMathOperator\Staircase{Staircase}
\DeclareMathOperator\InterReduce{InterReduce}
\DeclareMathOperator\fail{fail}
\DeclareMathOperator\shift{shift}
\DeclareMathOperator\supp{supp}
\journal{Journal of Symbolic Computation}

\begin{document}

\begin{frontmatter}



\title{In-depth comparison of the Berlekamp -- Massey -- Sakata and
the Scalar-FGLM algorithms:\\the non adaptive variants}

\author{J\'er\'emy
  Berthomieu\corref{cor1}}
\cortext[cor1]{Laboratoire d'Informatique de Paris~6,
  Universit\'e Pierre-et-Marie-Curie, bo\^ite courrier~169, 4~place
  Jussieu, F-75252 Paris Cedex~05, France.}
\ead{jeremy.berthomieu@lip6.fr}
\author{Jean-Charles Faug\`ere}
\ead{jean-charles.faugere@inria.fr}
\address{Sorbonne Universit\'es, \textsc{UPMC} Univ Paris~06,
  \textsc{CNRS}, \textsc{INRIA},\\
  Laboratoire d'Informatique de Paris~6 (\textsc{LIP6}),
  \'Equipe \textsc{PolSys},\\
  4 place Jussieu, 75252 Paris Cedex 05, France}

\begin{abstract}
  We compare thoroughly the
\textsc{Berlekamp~-- Massey~-- Sakata} algorithm and
the \textsc{Scalar-FGLM} algorithm, which compute both the ideal of
relations of a multidimensional linear recurrent sequence.

Suprisingly, their behaviors differ. We detail in which
way they do and prove that it is not possible to tweak one of the
algorithms in order to mimic exactly the behavior of the other.


\end{abstract}


\end{frontmatter}


\tableofcontents
\section{Introduction}\label{s:intro}
Computing the smallest linear recurrence relation satisfied by a
sequence is a fundamental problem in Computer Science. It is the
shortest linear feedback shift register (\textsc{LFSR}) which
generates the sequence. The length of this relation estimates the
linear complexity of the sequence.

In the 18th century, Gau\ss{} was interested in predicting the next
term of a sequence. Given a discrete set $(u_i)_{i\in\N}$, find the
best coefficients, in the least-squares sense,
$(\alpha_i)_{1\leq i\leq d}$
that will approximate $u_i$ by
$-\sum_{k=1}^d\alpha_k\,u_{i-k}$. Least-square sense means that the solution
minimizes the sum of the squares of the errors.

This problem has also been extensively used in Digital
Signal Processing theory
and applications.  Numerically, Levinson~-- Durbin recursion method
can be used to solve this problem. Hence,
to some extent, the original Levinson~-- Durbin problem in Norbert
Wiener's Ph.D. thesis, \cite{Levinson47,Wiener49}, predates the Hankel
interpretation of the \bm algorithm, see for instance~\cite{JoMa89}.

The \bm algorithm (\BM, \cite{Berl68,Mass69}) guesses a solution of
this problem for
sequences with one parameter, \ie in the one-dimensional case. This
algorithm has been tremendously studied and many variants were
designed.
We refer the reader
to~\cite{KaPa91,KaYa13,Kalto06} for a very nice classification
of the \BM algorithms for solving this problem, and for its
generalization to
matrix sequences.

Classically, two designs of the \BM algorithm are used.

The first one
assumes that the coefficients of the sequence $(u_i)_{i\in\N}$ are
given \emph{online}, \ie $u_{i+1}$ is known only
after $u_i$, and that a bound $d$ is given such that $u_{d-1}$ will be
computed but not $u_d$. Then, the \BM algorithm guesses a linear recurrence
relation satisfied by the table $(u_0,\ldots,u_i)$, checks if this
relation is satisfied by $(u_0,\ldots,u_{i+1})$
and updates the relation if it was
not. The algorithm stops when reaching $u_{d-1}$.

The other one 
assumes that the table
$(u_0,\ldots,u_{d-1})$ is known at
once. Then, the algorithm finds the kernel of
the Hankel matrix of size $\ceil{d/2}\times\floor{d/2}$
associated with the sequence
$(u_i)_{i\in\N}$. A complexity breakthrough is reached since this
comes down to calling the extended Euclidean
algorithm between $x^d$ and $U(x)=\sum_{i=0}^{i-1}u_i\,x^{d-i-1}$ and
stopping it prematurely when reaching a remainder of degree 
strictly less than $\floor{d/2}$. The relation is then given by the
B\'ezout coefficient
of $U(x)$ associated with this remainder. See~\cite{Blackburn1997,
Dornstetter1987}.

Sakata extended the \BM algorithm to $2$ dimensions in~\cite{Sakata88} and
then to $n$ dimensions in~\cite{Sakata90,Sakata09}. The so-called \bms
algorithm (\BMS) guesses a \gb of the ideal of relations
satisfied by the first terms of the input sequence, \cite[Lemma~5]{Sakata90}.

In a way, the \BMS algorithm extends the first design of the \BM
algorithm, as when calling the \BMS algorithm on a univariate
sequence, it behaves exactly like the \BM algorithm on this
sequence.

The so-called \sFGLM algorithm, presented
in~\cite{issac2015,berthomieu:hal-01253934} guesses the reduced
\gb of the ideal of relations
of a sequence. It extends the second design of the \BM algorithm
through the computation of the kernel of a \emph{multi-Hankel} matrix,
the multivariate generalization of a Hankel matrix. However, no fast
method is currently known for
computing this kernel.

While the second design of the \BM algorithm seems more efficient than
the first one,
mainly thanks to fast Euclidean algorithms, it is not
clear how their multidimensional extensions compare.
Surprisingly, the \BMS and the \sFGLM algorithm behave so differently that
it is not possible to apply a small modification on either algorithm
in order to simulate the behavior of the other.

\subsection{Related works}
Computing linear recurrence relations of multi-dimensional sequences
finds applications in Coding Theory, Computer Algebra and
Combinatorics.

Historically, the \BM algorithm was designed to decode cyclic codes,
like the \textsc{BCH}
codes, \cite{BoseRC1960,Hocquenghem1959}. Therefore,
decoding $n$-dimensional cyclic codes, a generalization of
Reed Solomon codes, was Sakata's motivation for designing the
\BMS algorithm in~\cite{Sakata91}.

On the other hand, as the output of the \BMS
and the \sFGLM algorithms is a \gb, a
natural application in Computer Algebra
is the computation of a \gb of an ideal for another order, typically
from a total degree ordering to an elimination ordering.
In fact the latest versions of the \spFGLM algorithm
rely heavily on the \BM and \BMS algorithms,
see~\cite{FM11,faugere:hal-00807540}.

Finally, computing linear
recurrence relations with \emph{polynomial} coefficients finds
applications in Computer Algebra for computing properties of
univariate and multivariate Special Functions. The Dynamic
Dictionary of Mathematical Functions (\textsc{DDMF}, \cite{DDMF}) generates
automatically web-pages on univariate special functions through the
differential equations they satisfy. Equivalently, they could be generated
through the linear recurrence relations satisfied by the sequence of
coefficients of their Taylor series. Deciding whether
\textsc{2D}/\textsc{3D}-space walks are D-finite or not finds
applications in Combinatorics,
see~\cite{BanderierF2002,BostanBMKM2014,BousquetMM2010,BousquetMP2003}. This
motivated the authors to extend the \sFGLM algorithm to
handle relations with polynomial coefficients
in~\cite{berthomieu:hal-01314266}.

\subsection{Contributions}
The main goal of this paper is to compare both the \BMS and the \sFGLM
algorithms. As it is not possible to store the whole input sequence,
both algorithms takes a bound as an input and only handle sequence
terms up to this index bound.

We start by recalling some classical notation
and definitions that shall be used in the proofs and the algorithms of
the paper in Section~\ref{s:prelim}.

Then, in order to be self-contained,
we dedicate the next two sections to a presentation of each algorithm.

A lot of articles, such
as~\cite{Bras-Amoros2006,Sakata88,Sakata90,Sakata09}, or book chapters,
such as~\cite[Chapter~10]{CoxLOS2015b}, present the \BMS
algorithm. Some of them deals with the
very general case of an
ordered domain. We specialize this description to the simpler case of
a polynomial ring
$\K[x_1,\ldots,x_n]$. In the \BMS algorithm, the input bound is a
monomial, so that the algorithm shall visit every monomial in
increasing order up to the bound.

On the other hand, in Section~\ref{s:Scalar-FGLM}, we describe the
\sFGLM algorithm with a point of view closer to the \BMS
algorithm. In the \sFGLM algorithm, the input bound is a set of terms
which contains the staircase of the computed \gb.

These presentations shall help us to first design a new algorithm in
between both of them in Section~\ref{s:New Solver}. 

Then, it will help us to compare them in
Section~\ref{s:comparison}, our main contribution of this paper.
We detail exactly how both algorithms behave similarly
and how, depending on the input, they can surprisingly differ.

A main likeness between both algorithms is that they determine which
monomials are in the \gb staircase. However, they handle the leading
terms outside of this staircase differently.

\begin{theorem}
  Let $\bu=(u_{i,j})_{(i,j)\in\N^2}$ be a sequence, let $\prec$ be a
  degree monomial ordering.
  
  Assuming we call each algorithm on $\bu$, $\prec$ and a bound
  allowing us to find the same set $S$ as the staircase, then
  \begin{itemize}
  \item for any monomial $m$ on the border of $S$, the \BMS algorithm
    returns a relation with leading term $m$.
    Therefore, the computed ideal of relations is zero-dimensional.
  \item the \sFGLM algorithm returns relations with leading terms on the
    border of $S$ but may fail to close the staircase. Therefore, the
    computed ideal of relations might be positive-dimensional.
  \end{itemize}
  If $\bu$ is linear recurrent and the bound big enough, then both
  algorithms compute correctly the ideal of relations of $\bu$.
\end{theorem}
The last part of the theorem is important as in most applications
$\bu$ \emph{is} linear recurrent. Therefore, both algorithms are able
to retrieve the ideal of relations of $\bu$.

We refer to Theorem~\ref{th:closed_staircase} for a more precise and general
version of this result.

By design, these algorithms return a set of relations, satisfied by
the sequence terms, and their shifts, \ie how far these relations have
been tested. The following theorem proves that the outputs of the
algorithms are quite different. This should convince the reader that
the algorithms do not compute the same thing whenever the bound is too
low or $\bu$ is not linear recurrent. It is a specialization
of Theorem~\ref{th:valid_shift} to the binomial sequence.

\begin{theorem}\label{th:intro_valid_shift}
  Let $\bin=\pare{\binom{i}{j}}_{(i,j)\in\N^2}$ be the sequence of the
  binomial coefficients and let $\prec$ be a total degree
  monomial ordering.

  Assuming we call each algorithm on $\bin$, $\prec$ and a bound
  allowing us to retrieve the same relations
  $x\,y-y-1,y^d,(x-1)^d$, with $d>2$.

  \begin{itemize}
  \item Then, the \sFGLM algorithm ensures that the shifts of the
    three relations are equal: they are still valid when multiplied by
    all the monomials of degree at most $d-1$.
  \item The \BMS algorithm ensures that the shifts of $y^d$ and
    $(x-1)^d$ are less than the shift of $x\,y-y-1$:
    relations $y^d$
    and $(x-1)^d$ are still valid when multiplied by all the
    monomials of degree at most $d-1$ while relation
    $x\,y-y-1$ is still valid when multiplied by all
    the monomials of degree at most $2\,d-3$.
  \end{itemize}
  In other words, the lesser the leading monomial of a relation
  computed by the \BMS algorithm, the
  greater its shift.
\end{theorem}

We mention earlier that the \spFGLM algorithm was a possible
application of these algorithms. Although, they are not meant to be
run with the lexicographical ordering, we prove the following result to
illustrate the difference in behaviors of these algorithms. This
result is extended to any dimension in
Theorem~\ref{th:shape_position}.

\begin{theorem}\label{th:intro_shape_position}
  Let $\bu=(u_{i,j})_{(i,j)\in\N^2}$ be a linear recurrent sequence whose
  ideal of relations $I=\langle g(y),x-f(y)\rangle$ is in
  \emph{shape position} for the $\LEX(y\prec x)$ ordering, with $\deg
  f<\deg g=d$ and $g$ squarefree.

  Assuming we call each algorithm on $\bu$, the $\LEX(y\prec x)$
  ordering, and a bound on the sequence terms.
  \begin{itemize}
  \item The \sFGLM algorithm, with the set of terms
    $T=\{1,y,\ldots,y^{d-1}\}$, yields the
    ideal $I$.
  \item The \BMS algorithm, visiting monomials $1,y,\ldots,y^d$, yields
    the ideal $\langle g(y),x\rangle$. This ideal is not $I$ unless $f=0$.
  \end{itemize}
  In other words, the \sFGLM algorithm can retrieve an ideal of
  relations in shape position while, in general, the \BMS algorithm cannot.
\end{theorem}

Finally, in Section~\ref{s:implem}, we compare the algorithms based on
the number of basic operations and the number of
table queries they perform.

We show that the \sFGLM
algorithm performs in general more queries to the table than the \BMS
algorithm. Yet, in the best case scenario where the leading terms
of the \gb of the ideal are all the monomials of a given degree, the
\sFGLM has a better behavior than the \BMS algorithm.

\subsection{Perspectives}
We are now in a position where the \BMS algorithm and the \sFGLM
algorithm are well understood and where we know that each
algorithm has strengths and weaknesses.

As anticipated in the original paper, the naive linear algebra solver
in the \sFGLM algorithm is its main weakness.
Therefore, a fast multi-Hankel solver could improve this algorithm.
Moreover, although its presentation is of a global
algorithm, it can be turned into an iterative one using naive Gaussian
elimination. Thus, a fast multi-Hankel arithmetic could also be useful for
an iterative variant of the algorithm.

On the other hand, the \BMS algorithm is a real iterative algorithm:
if in addition of
the relations, one outputs the
set of failing relations (see Remark~\ref{rk:choose_Sm}),
then one could continue the computation
up to a farther bound with no additional cost. Moreover, it is a faster
algorithm since it uses a polynomial arithmetic instead of a linear
algebra one.

A consequence of this paper could be the design of an hybrid algorithm
taking advantage of both the \BMS and the \sFGLM algorithms.
Another direction would be the study of adaptive variants of the algorithms.
The \asFGLM (\cite{issac2015,berthomieu:hal-01253934}) is a more
efficient variant of the \sFGLM
algorithm trying not to test too far the computed relations in order
to minimize the table queries and the complexity.
Likewise, one could design an adaptive variant of the \BMS algorithm based on
this philosophy and study their complexities.

In summary, the goal would be to take a step further
in the hybrid approach using the efficiency of the polynomial
arithmetic in the \BMS
algorithm to compute the relations and the smaller number of queries
performed by the \asFGLM algorithm.


\section{Preliminaries}\label{s:prelim}
In this section, we present classical notation that shall be used all
along the paper. We also present some definitions that will be useful
for all the proofs and algorithms.

\subsection{Sequences and relations}
Let $n\geq 1$, we write $\bi=(i_1,\dots,i_n)\in\N^n$. Likewise, we
denote $\x=(x_1,\ldots,x_n)$ and for $\bi\in\N^n$, we write
$\x^{\bi}=x_1^{i_1}\,\cdots\,x_n^{i_n}$.  Let
$\bu=(u_\bi)_{\bi\in\N^n}$ be a $n$-dimensional sequence over the
field $\K$.  If there exists a finite set of indices
$\cK\subset\N^n$ and numbers
$(\alpha_\bk)_{\bk\in\cK}$ in the field $\K$ such that
\begin{equation}
  \label{eq:recmulti}
  \forall \bi\in\N^n,\,\sum_{\bk\in\cK}
  \alpha_\bk\, u_{\bk+\bi}=0,
\end{equation}
then we say that $\bu$ satisfies the linear recurrence relation
(simply relation in the following) defined by
$\balpha=(\alpha_\bk)_{\bk\in\cK}$.

\begin{example}\label{ex:binom}
  Let $\bin$ be the 2-dimensional sequence of the binomial
  coefficients, $\bin = \left(\binom{i}{j}\right)_{(i,j)\in\N^2}$.
  Then the Pascal's rule:
  \[
  \forall (i,j)\in\N^2,\, \bin_{i+1,j+1}-\bin_{i,j+1}-\bin_{i,j}=0
  \]
  is a linear recurrence relation for the sequence $\bin$.
\end{example}

As we can only work with a finite number of terms of a sequence, in
this paper, a \emph{table} shall denote a finite subset of terms of a
sequence: it is one of the input parameters of the algorithms.

Given a finite table extracted from the sequence $\bu$, the main
purpose of the \BMS and the \sFGLM algorithms is to, lousy speaking,
determine a minimal set of relations that will allow us to generate
this finite table using only the values of $\bu$ on their supports.

In order to study the relations satisfied by the sequence $\bu$, it
will be useful to associate them with polynomials in
$\K[\x]$.

\begin{definition}
  Let $f=\sum_{\bk\in\cK}\alpha_\bk\,\x^\bk\in\K[\x]$.
  We will denote by $\cro{f}_{\bu}$, or $\cro{f}$ when no ambiguation
  arises, the linear combination
  $\sum_{\bk\in\cK}\alpha_\bk \,u_\bk$.
  Moreover, if $\balpha$ defines a relation for $\bu$, that is for all
  $\bi\in\N^n$, $\cro{\x^{\bi}\,f}=0$, then we say that $f$ is the
  polynomial of this relation.
\end{definition}
The main benefit of the $[\,]$ notation resides in the immediate fact
that for all index $\bi$,
$\left[\x^\bi\,f\right]=\sum_{\bk\in\cK}
\alpha_\bk\,u_{\bk+\bi}$.

In the previous example, the Pascal's
rule relation is associated with polynomial $P=x\,y-y-1$, so that
\[\forall (i,j)\in\N^2,\,[x^i\,y^j\,P]=0.\]

\begin{definition}[\cite{FitzpatrickN90,Sakata88}]~\label{def:lin_rec}
  Let $\bu=(u_{\bi})_{\bi\in\N^n}$ be a $n$-dimensional sequence
  with coefficients in $\K$. The sequence $\bu$ is \emph{linear
    recurrent} if from a nonzero finite number of initial terms
  $\{u_{\bi},\ \bi\in S\}$, and a finite number of linear recurrence
  relations, without any contradiction, one can compute any term of
  the sequence.
  
  Equivalently, $\bu$ is linear recurrent if its ideal of relations
  $\{f,\ \forall\,m\in\K[\x],\cro{m\,f}=0\}$ is \emph{zero-dimensional}.
\end{definition}

\subsection{\gbs}
We let $\cT$ be the set of all monomials
in $\K[\x]$, \ie $\cT=\{\x^{\bi},\ \bi\in\N^n\}$.
A monomial ordering $\prec$ on $\K[\x]$ is an order relation
satisfying the following three classical properties:
\begin{enumerate}
\item for all $m\in\cT$, $1\preceq m$;
\item for all $m,m',s\in\cT$, $m\prec m'\Rightarrow m\,s\prec m'\,s$;
\item every subset of $\cT$ has a least element for $\prec$.
\end{enumerate}

For a monomial ordering $\prec$ on $\K[\x]$, the
\emph{leading monomial} of $f$, denoted $\LM(f)$, is the
greatest monomial in the support of $f$ for $\prec$. The
\emph{leading coefficient} of $f$, denoted $\LC(f)$, is the
nonzero coefficient of $\LM(f)$. The \emph{leading term} of $f$,
$\LT(f)$, is defined as $\LT(f)=\LC(f)\,\LM(f)$.
For an ideal $I$, we denote, classically,
$\LM(I)=\{\LM(f),\ f\in I\}$.

We recall briefly the definition of a \gb and a staircase.
\begin{definition}
  Let $I$ be a nonzero ideal of $\K[\x]$ and let $\prec$ be a monomial ordering.
  A set $\cG\subseteq I$ is a \emph{\gb} of $I$ if for all $f\in I$,
  there exists $g\in\cG$ such that $\LM(g)|\LM(f)$.

  The set $\cG$ is a \emph{minimal} \gb of $I$ if for any $g\in\cG$,
  $\cG\setminus\{g\}$ does not span $I$.

  Furthermore, $\cG$ is (minimal) \emph{reduced} if for any $g,g'\in\cG$,
  $g\neq g'$ and any monomial $m\in\supp g'$, $\LT(g)\nmid m$.

  Let $\cG$ be a reduced truncated \gb, the \emph{staircase} of $\cG$
  is
  \[S=\Staircase(\cG)=\{s\in\cT,\ \forall\,g\in\cG, \LM(g)\nmid s\}.\]
  It is also the canonical basis of $\K[\x]/I$.
\end{definition}

\begin{remark}
  By definition, a staircase is stable by division: that is, for any
  $s,s'\in\cT$, if $s$ is in the staircase of $\cG$ and $s'|s$,
  then $s'$ is also in the staircase of $\cG$.

  In some instances, the goal will be to make the smallest \gb
  staircase from a monomial set $S$: this is done by adding all the
  divisors of the elements of $S$. We denote this by stabilizing $S$
  with the $\Stabilize(S)$ command.
\end{remark}

We, now, present notation
of~\cite{Sakata88,Sakata90,Sakata09} and relate it to polynomials
and polynomial ideals. This definition shall act like a dictionary
between Sakata's notation in these paper and the polynomials algebra
notation. We also refer to~\cite{Guisse2016}, \cite[Section~1]{Mora09}
and~\cite[Section~2]{Sakata09} for this kind
of dictionary.
\begin{definition}
  Given a set of polynomials $G\subseteq\K[\x]$.
  \begin{itemize}
  \item $\Sigma(G)=\{\x^{\bi},\ \exists\,g\in G,\ \LM(g)|\x^{\bi}\}$.

    Whenever, $G$ is a \gb of an ideal $I$, $\Sigma(G)$ is by
    definition $\LM(I)$.
  \item As $\Sigma(G)$ satisfies $\forall\,s\in\Sigma(G), m\in\cT$, if
    $s|m$, then $m\in\Sigma(G)$, it has minimal elements for the
    division.  They form the set $\sigma(G)=\min_|(\Sigma(G))$.

    Whenever, $G$ is a minimal \gb of an ideal $I$, $\sigma(G)$ is by
    definition $\LM(G)$.
  \item $\Delta(G)=\cT\setminus\Sigma(G)$.

    Whenever $G$ is a \gb of an ideal $I$, $\Delta(G)$ is its
    \emph{staircase}, the canonical basis of $\K[\x]/I$.
  \item As $\Delta(G)$ satisfies $\forall\,d\in\Delta(G), m\in\cT$, if
    $m|d$, then $m\in\Delta(G)$, it has maximal elements for the
    division.  They form the set $\delta(G)=\max_|(\Delta(G))$.

    Whenever, $G$ is a \gb of an ideal $I$, $\delta(G)$ is by
    definition the \emph{corner set} of the staircase.
  \end{itemize}
\end{definition}
\gb theory allows us to choose any monomial ordering $\prec$.
Among all the monomial ordering, we will mainly use the
\begin{itemize}
\item $\LEX(x_n\prec\cdots\prec x_1)$ ordering which compares
  monomials as follows $\x^{\bi}\prec\x^{\bi'}$ if, and only if, there
  exists $k$, $1\leq k\leq n$ such that for all $\ell<k$,
  $i_{\ell}=i_{\ell}'$ and $i_k<i_k'$,
  see~\cite[Chapter~2, Definition~3]{CoxLOS2015};
\item $\DRL(x_n\prec\cdots\prec x_1)$ order which compares monomials as
  follows $\x^{\bi}\prec\x^{\bi'}$ if, and only if,
  $i_1+\cdots+i_n<i_1'+\cdots+i_n'$ or $i_1+\cdots+i_n=i_1'+\cdots+i_n'$
  and there exists $k$,
  $2\leq k\leq n$ such that for all $\ell>k$, $i_{\ell}=i_{\ell}'$
  and $i_k>i_k'$. Equivalently, there exists $k$, $1\leq k\leq n$ such
  that for all $\ell>k$, $i_1+\cdots+i_{\ell}=i_1'+\cdots+i_{\ell}'$
  and $i_1+\cdots+i_k<i_1'+\cdots+i_k'$,
  see~\cite[Chapter~2, Definition~6]{CoxLOS2015}.
\end{itemize}

However, in the \BMS algorithm, we need to be able to enumerate all
the monomials up to a bound monomial. This forces the user to take an
ordering $\prec$
such that for all $M\in\cT$, the set $\{m\prec M,\ m\in\cT\}$
is finite. Such an ordering $\prec$ makes $(\N^n,\prec)$ isomorphic to
$(\N,<)$, thus it makes sense to speak about the next monomial
for $\prec$.

This request excludes for instance the $\LEX$ ordering, and more
generally any elimination ordering. In other words, only weighted
degree ordering, or \emph{weight ordering}, should be used. It is well
known that any monomial ordering $\prec$ on $\cT$ can be obtained
from a matrix $A\in\R^{n\times n}$ through: $\x^{\bi}\prec \x^{\bj}$
if, and only if, $\x^{A\cdot\bi}\prec_{\LEX}\x^{A\cdot\bj}$,
see~\cite{Erdos56}.
Such a matrix
$A\in\R^{n\times n}$ defines a monomial ordering if its first row
is nonnegative. It defines a weight ordering if its first row is
positive, see~\cite{Robbiano1986}
and~\cite[Chapter~2, Exercises~4.10 and~4.11]{CoxLOS2015}

\begin{definition}
  Let $I$ be a homogeneous ideal of $\K[\x]$ and let $\prec$ be a
  monomial ordering. A set $\cG\subseteq I$ is a
  \emph{$d$-truncated \gb}, or truncated \gb of $I$ up to degree
    $d$, if for all $g\in\cG$, $\deg g\leq d$ and for
    for all $f\in I$, if $\deg f\leq d$, then there exists a $g\in\cG$
    such that
    $\LT(g)|\LT(f)$.

    This can be computed using any \gb
    algorithm by discarding critical pairs of degree greater than $d$.
\end{definition}

For an affine ideal $I$, an analogous definition of $d$-truncated \gb
exists. It is the output of a \gb algorithm discarding all critical pairs 
$(f,f')$ with
$\deg\LT(f)+\deg\LT(f')-\deg\lcm(\LT(f),\LT(f'))> d$, \ie with
degree higher than $d$. In this situation, a
$d$-truncated \gb $\cG$  will span
the subspace of polynomials $\sum_{g\in\cG}h_g\,g$ with $\deg h_g\leq
d-\deg g$.

A truncated \gb $\cG$ is \emph{reduced} if for any $g,g'\in\cG$ and
any monomial $m\in\supp g$, $\LM(g')\nmid m$.

The following definition extends the definition of the staircase of a \gb to
truncated \gb.
\begin{definition}\label{def:staircase}
  Let $\cG$ be a reduced truncated \gb, the \emph{staircase} of $\cG$
  is
  \[S=\Staircase(\cG)=\{s\in\cT,\ \forall\,g\in\cG, \LM(g)\nmid s\}.\]
\end{definition}

\subsection{Gorenstein ideals}
From any ideal $J\subseteq\K[\x]$, it is clear that one can construct
a sequence $\bu=(u_{\bi})_{\bi\in\N^n}$ whose ideal of
relations contains $J$: from a \gb $\cG$ of $J$ and staircase $S$, set the
values of the sequence terms $u_{\bi}=[\x^{\bi}]$, for $\x^{\bi}\in S$,
as desired and then
computes the terms $u_{\bj}=[\x^{\bj}]$, for $\x^{\bj}\in\LM(I)$,
using the relations given by $\cG$.

However, Proposition~3.3 in~\cite{Brachat2010} proves that there are
nonzero ideals of $\K[\x]$ that cannot be the ideals
of relations of linear recurrent sequences, whenever $n\geq 2$.
Indeed, the ideal of relations is necessarily \emph{Gorenstein},
\cite{Gorenstein52, Macaulay34}, and problems occur only if $J$ has a zero of
multiplicity at least $2$.

For instance, there is no bivariate sequence
$\bu=(u_{i,j})_{(i,j)\in\N^2}$ whose ideal of relations $I$ is
$J=\langle x^2,x\,y,y^2\rangle$. That is, any sequence $\bu$ satisfying
$u_{2+i,j}=u_{1+i,1+j}=u_{i,2+j}=0$, for all $(i,j)\in\N^2$, satisfies a relation
induced by a degree-$1$ polynomial. Hence, $I$ strictly contains $J$.

The following theorem can also be found
in~\cite[Theorem~8.3]{ElkadiM2007}.
\begin{theorem}
  Let $I\subseteq\K[\x]$ be a $0$-dimensional ideal and let
  $R=\K[\x]/I$. The ideal $I$
  (\resp ring $R$) is \emph{Gorenstein} if equivalently
  \begin{enumerate}
  \item $R$ and its dual are isomorphic as $R$-modules;
  \item there exists a $\K$-linear form $\tau$ on $R$ such that the
    following bilinear form is non degenerate
    \begin{align*}
      R\times R&\to \K\\
      (a,b) &\mapsto \tau(a\,b).
    \end{align*}
  \end{enumerate}
\end{theorem}

On the one hand, this result is important for the
\spFGLM application. If the input ideal is not Gorenstein, the
output ideal will be bigger. However, this can be easily tested by
comparing the degrees of the input and output ideals. On the other
hand, this yields a probabilistic test for the Gorenstein property of
an ideal $J$. Pick at random initial conditions, construct a sequence
thanks to these initial conditions and $J$ and then compute the ideal
$I$ of relations of the sequence. If $I=J$, then $J$ is Gorenstein.
We refer to~\cite{DaleoH2016} for another test on the Gorenstein
property of an ideal.

\subsection{Multi-Hankel matrices}
A matrix $H\in\K^{m\times n}$ is \emph{Hankel}, if there exists a
sequence $\bu=(u_i)_{i\in\N}$ such that for all $(i,i')\in\N^n$, the
coefficient $h_{i,i'}$ lying on the $i$th row and $i'$th column of $H$
satisfies $h_{i,i'}=u_{i+i'}$.

In a multivariate setting, we can extend this Hankel matrices notion
to \emph{multi-Hankel} matrices. Indexing the rows and columns with
monomials $\x^{\bi}=x_1^{i_1}\,\cdots\,x_n^{i_n}$ and
$\x^{\bi'}=x_1^{i'_1}\,\cdots\,x_n^{i'_n}$, the
coefficient of $H$ lying on the row labeled with
$\x^{\bi}$ and column labeled with $\x^{\bi'}$ is
$u_{\bi+\bi'}$. Given two sets of monomials $U$ and $T$, we let
$H_{U,T}$ be the multi-Hankel matrix with rows (\resp columns) indexed
with monomials in $U$ (\resp $T$).
\begin{example}
  Let $\bu=(u_{i,j})_{(i,j)\in\N^2}$ be a sequence.
  \begin{enumerate}
  \item Let $U=\{1,y,y^2,x,x\,y,x\,y^2,x^2,x^2\,y,x^2\,y^2\}$ and
    $T=\{1,y,x,x\,y,x^2,x^2\,y\}$, then
    \[H_{U,T}=\kbordermatrix{
      &1&y&&x&x\,y&&x^2&x^2\,y\\
      1    &u_{0,0} &u_{0,1} &\vrule&u_{1,0} &u_{1,1} &\vrule&u_{2,0} &u_{2,1}\\
      y    &u_{0,1} &u_{0,2} &\vrule&u_{1,1} &u_{1,2} &\vrule&u_{2,1} &u_{2,2}\\
      y^2  &u_{0,2} &u_{0,3} &\vrule&u_{1,2} &u_{1,3} &\vrule&u_{2,2}
      &u_{2,3}\\
      \cline{2-9}
      x    &u_{1,0} &u_{1,1} &\vrule&u_{2,0} &u_{2,1} &\vrule&u_{3,0} &u_{3,1}\\
      x\,y &u_{1,1} &u_{1,2} &\vrule&u_{2,1} &u_{2,2} &\vrule&u_{3,1} &u_{3,2}\\
      x\,y^2 &u_{1,2} &u_{1,3} &\vrule&u_{2,2} &u_{2,3} &\vrule&u_{3,2} &u_{3,3}\\
      \cline{2-9}
      x^2    &u_{2,0} &u_{2,1} &\vrule&u_{3,0} &u_{3,1} &\vrule&u_{4,0} &u_{4,1}\\
      x^2\,y &u_{2,1} &u_{2,2} &\vrule&u_{3,1} &u_{3,2} &\vrule&u_{4,1} &u_{4,2}\\
      x^2\,y^2 &u_{2,2} &u_{2,3} &\vrule&u_{3,2} &u_{3,3} &\vrule&u_{4,2} &u_{4,3}
    }.
    \]
    We can see that this matrix is a $3\times 3$ \emph{block-Hankel}
    matrix with Hankel blocks of size $3\times 2$.
  \item Let $T=\{1,y,x,y^2,x\,y,x^2\}$, then the following matrix has
    a less obvious structure:
    \[H_{T,T}=\kbordermatrix{
      &1&y&x&y^2&x\,y&x^2\\
      1 &u_{0,0} &u_{0,1} &u_{1,0} &u_{0,2} &u_{1,1} &u_{2,0}\\
      y &u_{0,1} &u_{0,2} &u_{1,1} &u_{0,3} &u_{1,2} &u_{2,1}\\
      x &u_{1,0} &u_{1,1} &u_{2,0} &u_{1,2} &u_{2,1} &u_{3,0}\\
      y^2 &u_{0,2} &u_{0,3} &u_{1,2} &u_{0,4} &u_{1,3} &u_{2,2}\\
      x\,y &u_{1,1} &u_{1,2} &u_{2,1} &u_{1,3} &u_{2,2} &u_{3,3}\\
      x^2 &u_{2,0} &u_{2,3} &u_{3,0} &u_{2,2} &u_{3,3} &u_{0,4}
    }.
    \]
  \end{enumerate}
\end{example}


\section{The \BMS algorithm}\label{s:BMS}
As in~\cite{Guisse2016}, we specialize to $\K[\x]$ the presentation of
the \BMS algorithm
given in~\cite{Bras-Amoros2006}, \cite{CoxLOS2015b} and~\cite{Sakata09}
in the more general case of ordered domains.

\subsection{A Polynomial interpretation of the \BMS algorithm}
Given a table $\bu=(u_{\bi})_{\bi\in\N^n}$ and a weight
ordering $\prec$ for $\x$. We let
$\cT_0=\{0\}\cup\{\x^{\bi},\ \bi\in\N^n\}$
and extend $\prec$ (still denoted by $\prec$)
to $\cT_0$ with the convention that $0\prec 1$.

The goal is to iterate on a monomial $m$,
by only considering, at each step, the table
$(u_{\bi})_{\bi\in\{\bk,\ \x^{\bk}\preceq m\}}$. As we only know partially the
table $\bu$, we need to define some notions according to this partial
knowledge at step $m$.
\begin{definition}
  \label{def:shift}
  Let $m\in\cT_0$. Let $f\in\K[\x]$, we say that the relation $f$ is
  \emph{valid up to $m$}, whenever
  \[\forall t\in\cT_0,\,
  \LM(t\,f) \preceq m \Rightarrow [t\,f]=0.\]
  We thus define the \emph{shift} of
  $f$ as $\shift (f)=\frac{m}{\LM(f)}$.

  We say that the relation $f$ \emph{fails} at $m$ whenever
  \begin{align*}
    \forall t\in\cT_0,\,
    t\,f \prec m \Rightarrow [t\,f]&=0,\\
    \cro{\frac{m}{\LM(f)}\,f}&\neq 0.
  \end{align*}
  We define the \emph{fail} of $f$ as
  $\fail(f)=m$.
  If the relation $f$ never fails, that is for all $t\in\cT_0$, $[t\,f]=0$,
  then by convention $\fail(f)=\shift(f)=+\infty$.
\end{definition}

\begin{proposition}
  Let $\bu$ be a table and $f\in\K[x]$ such that $\fail(f)\succ m$. For all
  $g\in\K[\x]$, if $\LM(g\,f)\preceq m$, then $[g\,f]=0$.
\end{proposition}
The following proposition show how to combine two failing relations with the
same shift in order to obtain a new relation valid with a bigger shift.

\begin{proposition}
  \label{prop:augm_shift}
  Let $f_1$ and $f_2$ be two relations such that
  $v=\frac{\fail (f_1)}{\LM(f_1)}=\frac{\fail (f_2)}{\LM (f_2)}$ and
  $e_1=\cro{v\,f_1}$, $e_2=\cro{v\,f_2}$. Let $f$ be the nonzero polynomial
  $f_1 - \frac{e_1}{e_2}\,f_2$. Then, for $i\in\{1,2\}$,
  $\fail(f)\succ \fail(f_i)$, \ie
  $\frac{\fail(f)}{\LM(f)}\succ v$.
\end{proposition} 
\begin{proof}
  For any $c\in\K$ and any $\mu\in\K[\x]$ such that
  $\LM(g)\prec v$, 
  we have
  $[\mu\,(f_1+c\,f_2)]=[\mu\,f_1]+c\,[\mu\,f_2]=0$, hence
  $\fail (f_1+c\,f_2)\succeq \fail (f_i)$.
  
  It remains to prove that for a good choice of $c$, we have a strict
  inequality:
  as, $[v\,(f_1+c\,f_2)]=[v\,f_1]+c\,[v\,f_2]=e_1+c\,e_2$, it is clear that
  $[v\,f]=[v\,(f_1-\frac{e_1}{e_2}\,f_2)]=0$, so that
  $\fail (f)\succ v\,\LM(f)\succeq\fail (f_i)$.
\end{proof}

\begin{definition}\label{def:im}
  Using the same notation as in Definition~\ref{def:staircase}, we let
  \[I_m = \{f\in\K[\x],\ \fail\pare{f}\succ m\},\]
  and $\cG_m$ be the least elements for $\prec$ of $I_m$, it is
  a truncated \gb of $I_m$:
  \begin{align*}
    \cG_m &= \min_{\prec}\{g,\ g\in I_m\},\\
    S_m &= \Staircase(\cG_m).
  \end{align*}
\end{definition}

\begin{example}
  Let us go back to Example~\ref{ex:binom} with sequence
  $\bin = \left(\binom{i}{j}\right)_{(i,j)\in\N^2}$.
  Consider $\K[x,y]$ with the $\DRL(y\prec x)$ ordering,
  and $m=x^2$.
  \[
  \begin{ytableau}
    \none[y^2] & 0   \\
    \none[y]   & 0        & 1   \\
    \none[1]   & 1        & 1        &     *(green)1 \\  
    \none      & \none[1] & \none[x] & \none[x^2]
  \end{ytableau}
  \]
  From this table, on the one hand, we can deduce that
  \begin{itemize}
  \item since it is not identically $0$, there is no relation with
    leading monomial $1$ valid up to $x^2$, hence $1\in S_{x^2}$;
  \item since $[y+\alpha]=\alpha$ and
    $[x\,(y+\alpha)]=1+\alpha$, there is no relation with leading
    monomial $y$ valid up to $x\,y$ and thus $x^2$, hence
    $y\in S_{x^2}$;
  \item since $[y\,(x+\beta\,y+\alpha)]=1$, there is no relation with leading
    monomial $x$ valid up to $x\,y$ and thus $x^2$, hence $x\in S_{x^2}$.
  \end{itemize}
  On the other hand, we can check that
  \begin{itemize}
  \item since $[y^2]=0$, relation $y^2$ is valid up to $y^2$ and thus
    $x^2$, hence $y^2\in\cT\setminus S_{x^2}$;
  \item since $[x\,y-1]=0$, relation $x\,y-1$ is valid up to $x\,y$
    and thus $x^2$, hence $x\,y\in\cT\setminus S_{x^2}$;
  \item since $[x^2-x]=0$, relation $x^2-x$ is valid up to $x^2$,
    hence $x^2\in\cT\setminus S_{x^2}$.
  \end{itemize}
  Therefore, $S_{x^2} = \{1,y,x\}$,
  $\max_|(S_{x^2})=\{y,x\}$ and $\min_|(\cT\setminus S_{x^2})=\{y^2,x\,y,x^2\}$.
  This is summed up in the following diagram. 
  
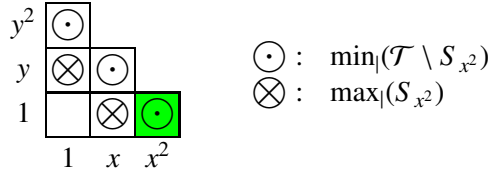
\[
  \begin{ytableau}
    \none[y^2] & \bigodot  \\
    \none[y]   & \bigotimes  & \bigodot\\
    \none[1]   &    & \bigotimes     &    *(green)\bigodot \\
    \none      & \none[1] & \none[x] & \none[x^2]
  \end{ytableau}
  \quad\quad
  \begin{array}{cl}
    \\\\
    \bigodot: &\min_|(\cT\setminus S_{x^2})\\
    \bigotimes: &\max_|(S_{x^2})
  \end{array}
  \]
  Let us notice that many relations with respective leading monomials
  $y^2,x\,y,x^2$ suit actually. These would be
  $y^2-\alpha_1\,x+\alpha_y\,y+\alpha_1,x\,y-(1+\alpha_1)\,x
  +\alpha_y\,y+\alpha_1$ and
  $x^2-(1+\alpha_1)\,x+\alpha_y\,y+\alpha_1$. Furthermore,
  $I_{x^2}$ is not stable by addition: $(x^2-x),(x^2-2\,x+1)\in I_{x^2}$ but
  $x^2-x-(x^2-2\,x+1)=(x-1)\not\in I_{x^2}$ since $\fail\pare{x-1}=x\,y$.
  Hence, $I_{x^2}$ is not an ideal of $\K[x,y]$.

  For $m=x^3$, with the following table, we find that
  
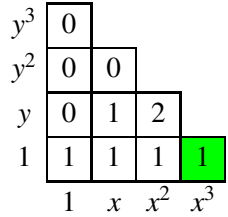
\[
  \begin{ytableau}
    \none[y^3] &0\\
    \none[y^2] & 0 &0  \\
    \none[y]   & 0        & 1   & 2\\
    \none[1]   & 1        & 1        & 1 &     *(green)1 \\  
    \none      & \none[1] & \none[x] & \none[x^2] &\none[x^3]
  \end{ytableau}
  \]
  \begin{itemize}
  \item since $[y^2]=[y\,y^2]=[x\,y^2]=0$, then $y^2$ is valid up to
    $x\,y^2$ and thus $x^3$;
  \item since $[x\,y-1]=[y\,(x\,y-1)]=0$ and $[x\,(x\,y-y)]=1$, then
    $x\,y-1$ fails at $x^2\,y$. Yet, since
    $[y]=[y\,y]=0$ and $[x\,y]=1$, then by
    Proposition~\ref{prop:augm_shift}, 
    $[x\,y-y-1]=[y\,(x\,y-y-1)]=0$ and $[x\,(x\,y-y-1)]$ vanishes as
    well. Hence, $x\,y-y-1$ is
    valid up to $x^2\,y$ and thus $x^3$;
  \item since $[x^2-x]=0$ and $[y\,(x^2-x)]=1$, then $x^2-x$ fails at
    $x^2\,y$. Likewise, since $[x-1]=0$ and $[y\,(x-1)]=1$, then
    $[x^2-2\,x+1]=0$ and $[y\,(x^2-2\,x+1)]=0$. Furthermore,
    $[x\,(x^2-2\,x+1)]=0$, so that
    $x^2-2\,x+1$ is valid up to $x^3$.
  \end{itemize}
  Therefore, $S_{x^3} = \{1,y,x\}$,
  $\max_|(S_{x^3})=\{y,x\}$ and $\min_|(\cT\setminus
  S_{x^3})=\{y^2,x\,y,x^2\}$.
  We can also check that these relations span the only valid relations with
  support in $S_{x^3}\cup\{y^2,x\,y,x^2\}$.
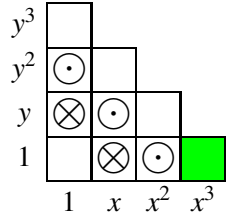
  \[
  \begin{ytableau}
    \none[y^3] &       \\
    \none[y^2] & \bigodot &  \\
    \none[y]   & \bigotimes  & \bigodot & \\
    \none[1]   &    & \bigotimes     &    \bigodot  & *(green)\\
    \none      & \none[1] & \none[x] & \none[x^2] & \none[x^3]\\
  \end{ytableau}
  \]
\end{example}

Although $I_m$ is not an ideal in general, we have the following results:

 \begin{proposition}\label{prop:Im-closed}
   Using the notation of Definitions~\ref{def:shift} and~\ref{def:im},
   \begin{enumerate}
   \item \label{eq:stab} $I_m$ is closed under multiplication by 
   elements of $\K[\x]$,
   \item for all monomials $t,t'$ such that $t| t'$,
     \begin{enumerate}
     \item \label{it:staircase}
       if $t'\in S_m$, then $t\in S_m$.
     \item \label{it:compl_staircase}
       if $t\in\cT\setminus S_m$, then $t'\in\cT\setminus S_m$,
     \end{enumerate}
   \end{enumerate}
 \end{proposition}

Moreover, it is clear that the sequence $(I_m)_{m\in\cT_0}$ is decreasing 
and that if $\bu$ is linear recurrent then $I = \bigcap_{m\in\cT_0}I_m$.
Therefore, $\left(S_m \right)_{m\in\cT_0}$
is increasing and its limit is $S$ the finite target staircase. Hence,
for $m$ big enough, $S_m$
will be the target staircase. We will give an upper bound in 
Proposition~\ref{prop:upperbound}.

The following result gives an intrinsic characterization of 
$S_m$ that is key in the iteration of the \BMS algorithm.

\begin{proposition}\label{prop:iter}
  For all monomial $m\in\cT_0$, 
  $S_m = \acc{\frac{\fail(f)}{\LM(f)},\ f\notin I_m}$.

  Furthermore, let $m^+$ be the successor of $m$. Let $s$ be a
  monomial in the staircase $S_{m^+}$. Then, $s$
  was added at step $m^+$, \ie $s\notin S_m$, if, and only if,
  $s|m^+$ and $\frac{m^+}{s}\in S_{m^+}\setminus S_m$.
\end{proposition}
\begin{proof}
  We shall prove the first assertion by double inclusion.
  If $s=\frac{\fail(f)}{\LM(f)}$ then for all $g\in\K[\x]$ such that
  $\LM(g)=s$, $\fail(g)\preceq m$, hence $s\notin\LM(I_m)$, $s\in S_m$.

  The reverse inclusion is proved by induction on $m$. For $m=0$,
  $S_m=\emptyset$ and there is nothing to do. Let us
  assume the inclusion
  is satisfied for a monomial $m$.
  
  Let $s\in S_{m^+}$. On the one hand, if $s\in S_m$, then
  there exists
  $f\in\K[\x]\setminus I_m\subseteq\K[\x]\setminus I_{m^+}$
  such that $s=\frac{\fail(f)}{\LM(f)}$.

  If, on the other hand, $s\in S_{m^+}\setminus S_m$, then there
  exists a relation $f\in\K[\x]$ such that $\LM(f)=s$,
  and $m\prec \fail(f) \preceq m^+$, hence $\fail(f)=m^+$ and $s$
  divides $m^+$.

  Let us assume that for all $g\in\K[\x]$ with $\LM(g)=\frac{m^+}{s}$, we have
  $\fail (g)\preceq m\prec m^+$. Therefore, $\frac{m^+}{s}\in S_m$ and
  there exists $h\notin I_m$ such that
  $\frac{\fail(h)}{\LM(h)}=\frac{m^+}{s}$. By
  Proposition~\ref{prop:augm_shift}, there is $\alpha\in\K$ such that
  $\fail(f-\alpha\,h)\succ m^+$. Since $\fail(h)\preceq m\prec
  m^+$, then $\LM(h)\preceq s$ and $\LM(f-\alpha\,h)=s$, hence
  $\frac{\fail(f-\alpha\,h)}{\LM(f-\alpha\,h)}\succ\frac{m^+}{s}$. This
  contradicts the fact that $\frac{m^+}{s}\in S_m$. Thus there exists
  a $g\in\K[\x]$ with $\LM(g)=\frac{m^+}{s}$ and $\fail(g)\succeq m^+$.
 
  Let $g$ be such a relation,
  since $\fail(f)=m^+$, then $[g\,f]\neq0$ and
  $\fail(g)=m^+$. Therefore,
  $\frac{\fail(g)}{\LM(g)}=\frac{m^+}{m^+/s}=s$
  so that $s\in\acc{\frac{\fail(f)}{\LM(f)},\ f\notin I_{m^+}}$.

  Now, we proved that $s\in S_{m^+}\setminus S_m$ implies $s|m^+$ and 
  $\frac{m^+}{s}\in S_{m^+}\setminus S_m$. This implication is clearly
  an equivalence.
\end{proof}

From this proposition it follows that
if $m\in\cT_0$, and if $m^+$
is its successor:
\begin{equation}\label{eq:recdelta}
  \max_|(S_{m^+}) = \max_{|}\pare{\max_|(S_m) \cup 
    \acc{\frac{m^+}{s},\ s\in\min_|(\cT\setminus S_m)\cap S_{m^+}}}
\end{equation}

Relation~\ref{eq:recdelta} allows us
to construct, iterating on the monomial $m$, 
the set of relations $G_m$ representing the truncated \gb of $I_m$.
Relations $g\in G_m$ are indexed by their leading
monomials, describing $\cT\setminus S_m$. 

\begin{remark}\label{rk:choose_Sm}
  We can also construct another set, describing the edge of $S_m$,
  still denoted $S_m$, as there is a one-to-one correspondence between a
  staircase and its edge. The relations $h\in S_m$ are indexed by their ratio
  $\frac{\fail(h)}{\LM(h)}$ between their fail and their leading
  monomial, describing the full staircase of $I_m$.

  When two relations $h$ and $h'$ in $S_m$ are such that
  $\frac{\fail(h)}{\LM(h)}=\frac{\fail(h')}{\LM(h')}$, then we only
  need to keep one. Since the goal is to combine a relation of
  $S_m$ with a relation failing at $m^+$ to make a new one with a bigger
  shift, as in Proposition~\ref{prop:augm_shift}, it is best to handle
  smaller polynomials.
\end{remark}
This yields Algorithm~\ref{algo:bms}.

\begin{algorithm2e}[htbp!]\label{algo:bms}
  \small
  \DontPrintSemicolon
  \TitleOfAlgo{The \BMS algorithm.}
  \KwIn{A table $\bu=(u_{\bi})_{\bi\in\N^n}$ with coefficients in
    $\K$, a monomial ordering $\prec$ and a monomial $M$
    as the stopping condition.}
  \KwOut{A set $G$ of relations generating $I_M$.}
  $T := \{m\in\K[\x],\ m\preceq M\}$.\tcp*{ordered for $\prec$}
  $G := \{1\}$.\tcp*{the future \gb}
  $S := \emptyset$.\tcp*{staircase edge, elements will be
    $[h,\fail(h)/\LM(h)]$} 
  
  \Forall{$m\in T$}{
    $S' := S$.\;
    \For{$g\in G$}{
      \If{$\LM(g)| m$}{
        $e:=\cro{\frac{m}{\LM(g)}\,g}_{\bu}$.\;
        \If{$e\neq 0$}{
          $S':=S'\cup\acc{\cro{\frac{g}{e},\frac{m}{\LM(g)}}}$.\;
        }
      }
    }
    $S':=\min_|\acc{[h,\fail(h)/\LM(h)]}$.
    \tcp*{see Remark~\ref{rk:choose_Sm}}
    $G':= \Border (S')$.\;
    
    \For{$g'\in G'$}{
      Let $g\in G$ such that $\LM(g)|\LM(g')$.\; 
      \uIf{$\LM(g) \nmid m$}{
        $g':=\frac{\LM(g')}{\LM(g)}\,g$.\tcp*{translates the relation}
      }
      \uElseIf{$\exists\,h\in S,
        \frac{m}{\LM(g')} |\fail(h)$}{
        $g':= \frac{\LM(g')}{\LM(g)}\,g
        -\cro{\frac{m}{\LM(h)}\,h}_{\bu}\,
        \frac{\LM(g')\,\fail(h)}{m}\,h$.\tcp*{see
          Proposition~\ref{prop:augm_shift}}
      }
      \lElse{
        $g':=g$.
      }
    }
    $G := G'$.\;
    $S := S'$.\;
  }  
  \KwRet $G$.
\end{algorithm2e}

We saw that for $m$ big enough, $S_m$
will be the target staircase. We now
give an upper bound.

\begin{proposition}\label{prop:upperbound}
  Let $\bu$ be a linear recurrent sequence and $I$ be its ideal of
  relations.

  Let $S$ be the
  staircase of $I$ for $\prec$. Let $s_{\max}$
  be the largest monomial in
  $S$. Then, for $m\succeq (s_{\max})^2$,
  $S_m = S$.

  Let $\cG$ be a minimal \gb of $I$ for $\prec$ and let $g_{\max}$ be the
  largest leading monomial of  $\cG$. Then, for $m\succeq
  s_{\max}\cdot\max_{\prec}(g_{\max},s_{\max})$,
  the \BMS algorithm returns a minimal \gb of $I$
  for $\prec$.
\end{proposition}
\begin{example}
  For the $\DRL(y\prec x)$ ordering, $I=\langle x^p,y^q\rangle$ and
  $q>p\geq 1$, we have, $s_{\max}=x^{p-1}\,y^{q-1}$ and
  $g_{\max}=y^q$. Therefore, the right staircase is found at most at step
  $m=x^{2\,p-2}\,y^{2\,q-2}$, while the \gb is found at most at step
  $x^{p-1}\,y^{q-1}\,\max_{\prec}(x^{p-1}\,y^{q-1},x^q)$, \ie
  $y^{2\,q-1}$ if $p=1$ and $x^{2\,p-2}\,y^{2\,q-2}$ otherwise.
\end{example}

From Propositions~\ref{prop:iter} and~\ref{prop:upperbound}, we can
deduce that $S=\acc{\frac{\fail(f)}{\LM(f)},\ f\notin I}$.%

\begin{example}\label{ex:binom_bms}
  We give the trace of the algorithm called on the binomial sequence
  $\bin$ for the
  $\DRL(y\prec x)$ ordering up to monomial $x^3$ (hence visiting
  all the monomials of degree at most $3$).

  To simplify the reading, whenever a relation succeeds in $m$ or
  cannot be tested in $m$, we
  skip the updating part as this relation remains the same.

  We start with the empty staircase $S$ and
  the relation $G=\{1\}$.
  \begin{enumerate}
  \item[] For the monomial $1$
    \begin{enumerate}
    \item[] The relation $g_1=1$ fails since $[\bin_{0,0}]=1$. Thus
      $S'=\{[1,1]\}$.
    \item[] $S'$ is updated to $\{[1,1]\}$ and $G'=\{y,x\}$.
    \item[] For the relation $g_1'=y$, $y\nmid 1$ thus $g_1'=y$.
    \item[] For the relation $g_2'=x$, $x\nmid 1$ thus $g_2'=x$.
    \item[] We update $G:=G'=\{y,x\}$ and $S:=S'=\{[1,1]\}$.
    \end{enumerate}
  \item[] For the monomial $y$
    \begin{enumerate}
    \item[] The relation $g_1=y$ succeeds since $[\bin_{0,1}]=0$.
    \item[] Nothing must be done for the relation $g_2=x$.
    \item[] $S'$ is set to $\{[1,1]\}$ and $G'=\{y,x\}$.
    \item[] We set $g_1'=y$ and $g_2'=x$.
    \item[] We update $G:=G'=\{y,x\}$ and $S:=S'=\{[1,1]\}$.
    \end{enumerate}
  \item[] For the monomial $x$
    \begin{enumerate}
    \item[] Nothing must be done for the relation $g_1=y$.
    \item[] The relation $g_2=x$ fails since $[\bin_{1,0}]=1$. Thus
      $S'=\{[1,1],[x,1]\}$.
    \item[] $S'$ is set to $\{[1,1]\}$ and $G'=\{y,x\}$.
    \item[] We set $g_1'=y$.
    \item[] For the relation $g_2'=x$, $x| x$ and
      $\frac{x}{x}|\fail(1)$, hence $g_2'=x-1$.
    \item[] We update $G:=G'=\{y,x-1\}$ and $S:=S'=\{[1,1]\}$.
    \end{enumerate}
  \item[] For the monomial $y^2$
    \begin{enumerate}
    \item[] The relation $g_1=y$ succeeds since $[\bin_{0,2}]=0$.
    \item[] Nothing must be done for the relation $g_2=x-1$.
    \item[] $S'$ is set to $\{[1,1]\}$ and $G'=\{y,x\}$.
    \item[] We set $g_1'=y$ and $g_2'=x-1$.
    \item[] We update $G:=G'=\{y,x-1\}$ and $S:=S'=\{[1,1]\}$.
    \end{enumerate}
  \item[] For the monomial $x\,y$
    \begin{enumerate}
    \item[] The relation $g_1=y$ fails since $[\bin_{1,1}]=1$. Thus
      $S'=\{[1,1],[y,x]\}$.
    \item[] The relation $g_2=x-1$ fails since $[\bin_{1,1}-\bin_{0,1}]=1$. Thus
      $S'=\{[1,1],[y,x],[x-1,y]\}$.
    \item[] $S'$ is set to $\{[y,x],[x-1,y]\}$ and $G'=\{y^2,x\,y,x^2\}$.
    \item[] For the relation $g_1'=y^2$, $y^2\nmid x\,y$ thus $g_1'=y^2$.
    \item[] For the relation $g_2'=x\,y$, $x\,y| x\,y$
      and $\frac{x\,y}{x\,y}|\fail(y)$, hence $g_2'=x\,y-1$.
    \item[] For the relation $g_3'=x^2$, $x^2\nmid x\,y$ thus $g_3'=x^2-x$.
    \item[] We update $G:=G'=\{y^2,x\,y-1,x^2-x\}$ and
      $S:=S'=\{[y,x],[x-1,y]\}$.
    \end{enumerate}
  \item[] For the monomial $x^2$
    \begin{enumerate}
    \item[] Nothing must be done for the relation $g_1=y^2$.
    \item[] Nothing must be done for the relation $g_2=x\,y-1$.
    \item[] The relation $g_3=x^2-x$ succeeds since
      $[\bin_{2,0}-\bin_{1,0}]=0$.
    \item[] $S'$ is set to $\{[y,x],[x-1,y]\}$ and $G'=\{y^2,x\,y,x^2\}$.
    \item[] We set $g_1'=y^2$, $g_2'=x\,y-1$ and $g_3'=x^2-x$.
    \item[] We update $G:=G'=\{y^2,x\,y-1,x^2-x\}$ and
      $S:=S'=\{[y,x],[x-1,y]\}$.
    \end{enumerate}
  \item[] For the monomial $y^3$
    \begin{enumerate}
    \item[] The relation $g_1=y^2$ succeeds since $[\bin_{0,3}]=0$.
    \item[] Nothing must be done for the relation $g_2=x\,y-1$.
    \item[] Nothing must be done for the relation $g_3=x^2-x$.
    \item[] $S'$ is set to $\{[y,x],[x-1,y]\}$ and $G'=\{y^2,x\,y,x^2\}$.
    \item[] We set $g_1'=y^2$, $g_2'=x\,y-1$ and $g_3=x^2-x$.
    \item[] We update $G:=G'=\{y^2,x\,y-1,x^2-x\}$ and
      $S:=S'=\{[y,x],[x-1,y]\}$.
    \end{enumerate}
  \item[] For the monomial $x\,y^2$
    \begin{enumerate}
    \item[] The relation $g_1=y^2$ succeeds since $[\bin_{1,2}]=0$.
    \item[] The relation $g_2=x\,y-1$ succeeds since $[\bin_{1,2}-\bin_{0,1}]=0$.
    \item[] Nothing must be done for the relation $g_3=x^2-x$.
    \item[] $S'$ is set to $\{[y,x],[x-1,y]\}$ and $G'=\{y^2,x\,y,x^2\}$.
    \item[] We set $g_1'=y^2$, $g_2'=x\,y-1$ and $g_3=x^2-x$.
    \item[] We update $G:=G'=\{y^2,x\,y-1,x^2-x\}$ and
      $S:=S'=\{[x,y],[y,x-1]\}$.
    \end{enumerate}
  \item[] For the monomial $x^2\,y$
    \begin{enumerate}
    \item[] Nothing must be done for the relation $g_1=y^2$.
    \item[] The relation $g_2=x\,y-1$ fails since
      $[\bin_{2,1}-\bin_{1,0}]=1$. Thus
      $S'=\{[y,x],[x-1,y],[x\,y-1,x]\}$.
    \item[] The relation $g_3=x^2-x$ fails since
      $[\bin_{2,1}-\bin_{1,1}]=1$. Thus
      $S'=\{[y,x],[x-1,y],[x\,y-1,x],[x^2-x,y]\}$.
    \item[] $S'$ is set to $\{[y,x],[x-1,y]\}$ and $G'=\{y^2,x\,y,x^2\}$.
    \item[] We set $g_1'=y^2$.
    \item[] For the relation $g_2'=x\,y$, $x\,y| x^2\,y$ and
      $\frac{x^2\,y}{x\,y}|\fail(y)$, hence $g_3'=x\,y-y-1$.
    \item[] For the relation $g_3'=x^2$, $x^2| x^2\,y$ and
      $\frac{x^2\,y}{x^2}|\fail(x-1)$, hence $g_3'=x^2-2\,x+1$.
    \item[] We update $G:=G'=\{y^2,x\,y-y-1,x^2-2\,x+1\}$ and
      $S:=S'=\{[y,x],[x-1,y]\}$.
    \end{enumerate}
  \item[] For the monomial $x^3$
    \begin{enumerate}
    \item[] Nothing must be done for the relation $g_1=y^2$.
    \item[] Nothing must be done for the relation $g_2=x\,y-y-1$.
    \item[] The relation $g_3=x^2-2\,x+1$ succeeds since
      $[\bin_{3,0}-2\,\bin_{2,0}+\bin_{1,0}]=0$.
    \item[] $S'$ is set to $\{[y,x],[x-1,y]\}$ and $G'=\{y^2,x\,y,x^2\}$.
    \item[] We set $g_1'=y^2$, $g_2'=x\,y-y-1$ and $g_3=x^2-2\,x+1$.
    \item[] We update $G:=G'=\{y^2,x\,y-y-1,x^2-2\,x+1\}$ and
      $S:=S'=\{[y,x],[x-1,y]\}$.
    \end{enumerate}
  \item[] The algorithm returns relations $y^2,x\,y-y-1,x^2-2\,x+1$,
    all three with a shift $x$.
  \end{enumerate}
\end{example}

\subsection{A Linear Algebra interpretation of the \BMS algorithm}
In order to make the presentation of the \BMS algorithm closer to that
of the \sFGLM algorithm, we propose to replace every evaluation using
the $[\,]$ operator with a matrix-vector product.

As stated above, given a monic relation $f=\LM(f)+\sum_{s\in S}\alpha_s\,s$,
testing the shift of this relation by a monomial $m$ is done with the
bracket operator, \ie testing whether $[m\,f]=0$ or not. Denoting
$\vec{f}$, the vector
\[
\vec{f}=\kbordermatrix{
  &1\\
  \vdots &\vdots\\
  s\in S &\alpha_s\\
  \vdots &\vdots\\
  \LM(f) &1
},
\] this can also
be done through testing if the following matrix-vector product
\[H_{m,S\cup\{\LM(f)\}}\,\vec{f}=
\kbordermatrix{
  &\cdots &s\in S&\cdots &\LM(f)\\
  m &\cdots &[m\,s] &\cdots &[m\,\LM(f)]
}\,
\begin{pmatrix}
  \vdots\\\alpha_s\\\vdots\\1
\end{pmatrix}
=0
\]
or not. In this setting, the definitions of the \emph{shift} and the
\emph{fail} of a relation, \ie Definition~\ref{def:shift}, become as follows.
\begin{definition}\label{def:fail_linal}
  Let $f=\LT(f)+\sum_{s\in S}\alpha_s\,s$ be a polynomial.

  The monomial $m$ is a \emph{shift of $f$} if
  \[H_{\{1,\ldots,m\},S\cup\{\LM(f)\}}\,\vec{f}=
  \kbordermatrix{
    &\cdots &s\in S&\cdots &\LM(f)\\
    1 &\cdots &[s] &\cdots &[\LM(f)]\\
    \vdots &&\vdots &&\vdots\\
    m &\cdots &[m\,s] &\cdots &[m\,\LM(f)]\\
  }\,
  \begin{pmatrix}
    \vdots\\\alpha_s\\\vdots\\1
  \end{pmatrix}
  =
  \begin{pmatrix}
    0\\\vdots\\0
  \end{pmatrix}.
  \]

  Let $m^+$ be the successor of $m$, $m^+\,\LM(f)$ is the \emph{fail of $f$} if
  \[H_{\{1,\ldots,m,m^+\},S\cup\{\LM(f)\}}\,\vec{f}=
  \kbordermatrix{
    &\cdots &s\in S&\cdots &\LM(f)\\
    1 &\cdots &[s] &\cdots &[\LM(f)]\\
    \vdots &&\vdots &&\vdots\\
    m &\cdots &[m\,s] &\cdots &[m\,\LM(f)]\\
    m^+ &\cdots &[m^+\,s] &\cdots &[m^+\,\LM(f)]
  }\,
  \begin{pmatrix}
    \vdots\\\alpha_s\\\vdots\\1
  \end{pmatrix}
  =
  \begin{pmatrix}
    0\\\vdots\\0\\e
  \end{pmatrix},
  \]
  with $e\neq 0$.
\end{definition}
We can also write another proof of Proposition~\ref{prop:augm_shift}
with a matrix viewpoint.
\begin{proof}[Proof of Proposition~\ref{prop:augm_shift}]
  Let $f_1=\LM(f_1)+\sum_{s\in S}\alpha_s\,s$ and
  $f_2=\LM(f_2)+\sum_{s\in S'}\beta_s\,s$ be monic. Let $v^-$ be the predecessor
  of $v$. Let $\tilde{S}=S\cup S'\setminus\{\LM(f_2),\LM(f_1)\}$,
  assuming $\LM(f_2)\neq\LM(f_1)$, then
  we have
  \begin{align*}
    H_{\{1,\ldots,v^-,v\},\tilde{S}\cup\{\LM(f_2),\LM(f_1)\}}\,(\vec{f_1}+c\,\vec{f_2})
    &=
      \begin{pmatrix}
        0\\\vdots\\0\\e_1+c\,e_2
      \end{pmatrix}\\
    \kbordermatrix{
    &\cdots &s\in \tilde{S}&\cdots &\LM(f_2) &\LM(f_1)\\
    1 &\cdots &[s] &\cdots &[\LM(f_2)] &[\LM(f_1)]\\
    \vdots &&\vdots &&\vdots&\vdots\\
    v^- &\cdots &[v^-\,s] &\cdots &[v^-\,\LM(f_2)] &[v^-\,\LM(f_1)]\\
    v &\cdots &[v\,s] &\cdots &[v\,\LM(f_2)] &[v\,\LM(f_1)]\\
    }\,
    \begin{pmatrix}
      \vdots\\\alpha_s+c\,\beta_s\\\vdots\\c\\1
    \end{pmatrix}
    &=
      \begin{pmatrix}
        0\\\vdots\\0\\e_1+c\,e_2
      \end{pmatrix}.
  \end{align*}
  It is now clear that vector $\vec{f}_1-\frac{e_1}{e_2}\,\vec{f}_2$
  is in the kernel of
  this matrix. That is, polynomial $f_1-\frac{e_1}{e_2}\,f_2$ has a
  shift $v$.
\end{proof}

Changing every evaluation into a matrix-vector product in the \BMS
algorithm yields 
the following presentation of the \BMS algorithm, namely
Algorithm~\ref{algo:bms_linalg}.

\begin{algorithm2e}[htbp!]\label{algo:bms_linalg}
  \small
  \DontPrintSemicolon
  \TitleOfAlgo{Linear Algebra variant of the \BMS algorithm.}
  \KwIn{A table $\bu=(u_{\bi})_{\bi\in\N^n}$ with coefficients in
    $\K$, a monomial ordering $\prec$ and a monomial $M$
    as the stopping condition.}
  \KwOut{A set $G$ of relations generating $I_M$.}
  $T := \{m\in\K[\x], m\preceq M\}$.\tcp*{ordered for $\prec$}
  $G := \{1\}$.\tcp*{the future \gb}
  $S := \emptyset$.\tcp*{staircase edge, elements will be
    $[h,\fail(h)/\LM(h)]$} 
  
  \Forall{$m\in T$}{
    $S' := S$.\;
    \For{$g\in G$}{
      \If{$\LM(g)| m$}{
        $e:=H_{\acc{\frac{m}{\LM(g)}},\supp(g)}\,\vec{g}$.\;
        \If{$e\neq 0$}{
          $S':=S'\cup
          \acc{\cro{\frac{g}{e},\frac{m}{\LM(g)}}}$.\;
        }
      }
    }
    $S':=
    \min_{\fail(h)\in S'}\acc{[h,\fail(h)/\LM(h)]}$.
    \tcp*{see Remark~\ref{rk:choose_Sm}}
    $G':= \Border (S')$.\;
    
    \For{$g'\in G'$}{
      Let $g\in G$ such that $\LM(g)|\LM(g')$.\; 
      \uIf{$\LM(g) \nmid m$}{
        $g':=\frac{\LM(g')}{\LM(g)}\,g$.\tcp*{shifts the relation}
      }
      \uElseIf{$\exists\,h\in S,
        \frac{m}{\LM(g')} |\fail(h)$}{
        $g':= \frac{\LM(g')}{\LM(g)}\,g
        -\pare{H_{\acc{\frac{m}{\LM(h)}},\supp(h)}\,\vec{h}}\,
        \frac{\LM(g')\,\fail(h)}{m}\,h$.\tcp*{see Prop.~\ref{prop:augm_shift}}
      }
      \lElse{
        $g':=g$.
      }
    }
    $G := G'$ \;
    $S := S'$ \;
  }  
  \KwRet $G$.
\end{algorithm2e}


\section{The \sFGLM algorithm}\label{s:Scalar-FGLM}
This section is devoted to the description of the \sFGLM algorithm
introduced in~\cite{issac2015,berthomieu:hal-01253934}.

The \sFGLM algorithm aims at computing linear recurrence relations of
a multidimensional sequence  with a matrix viewpoint and
an approach close to the \FGLM algorithm, see~\cite{FGLM93}.

The main idea is to shift the linear recurrence relations in order to
determine their coefficients.
As we can only know a finite number of the sequence terms, we need the
following definition.
\begin{definition}\label{def:shift2}
  Let $f\in\K[\x]$ and $T$ be a set of monomials in $\x$, we say that
  $f$ has a shift $T$ if
  \begin{equation}
    \label{eq:shift2}
    \forall m\in T,\ [m\,f]=0.
  \end{equation} 
\end{definition}

\begin{remark}
  We would like to emphasize that this definition is close to
  Definition~\ref{def:shift} of the shift for the \BMS algorithm.

  Whenever $T$ is a set of monomials $T_M=\{m,\ m\preceq M\}$,
  $f$ has a shift $T_M$ if, and only if,
  $f$ has a shift $M$, \ie $\fail(f)\succ \LM (f)\,M$.

  Unless stated otherwise, we will now always assume that the set $T$ is
  stable by division.
\end{remark}

From the relations
$\cro{\x^{\bi+\bd}+\sum_{\bk\in\cK}\alpha_{\bk}\,\x^{\bi+\bk}}=0$, valid for
all $\x^{\bi}\in T$, one can deduce that the polynomial
$P=\x^{\bd}+\sum_{\bk\in\cK}\alpha_{\bk}\,\x^{\bk}$ satisfies
$[m\,P]=0$ for all $m\in T$. In other words, $P$ has a shift $T$.

To determine $P$ with a shift $T_M$, it suffices to solve the linear system
\[
\begin{cases}
  \cro{\x^{\bd}+\sum_{\bk\in\cK}\alpha_{\bk}\,\x^{\bk}}&=0\\
  \quad\vdots &\phantom{=}  \ \vdots\\
  \cro{m\,x^{\bd}+\sum_{\bk\in\cK}\alpha_{\bk}\,m\,\x^{\bk}}&=0\\
  \quad\vdots &\phantom{=}  \ \vdots\\
  \cro{M\,x^{\bd}+\sum_{\bk\in\cK}\alpha_{\bk}\,M\,\x^{\bk}}&=0.
\end{cases}
\]

Before determining the coefficients of the relations, one needs to
determine their support.

\begin{definition}\label{def:useful_staircase}
  Let $T$ be a finite subset of terms. We say that a finite set
  $S\subset T$ is a \emph{useful staircase} with respect to $\bu$, $T$
  and $\prec$ if
  \[\sum_{t\in S}\beta_t\,[m\,t]=0, \ \ \forall\,m\in S\]
  implies that $\beta_t=0$ for all $t\in S$, $S$ is maximal for
  the inclusion and minimal for $\prec$. We compare two ordered sets
  for $\prec$ by seeing them as tuples of their elements and then
  comparing them lexicographically.
\end{definition}

We recall that for two sets of terms $U$ and $T$, the multi-Hankel
matrix associated with $U$ and $T$ is
\[H_{U,T} =
\kbordermatrix{
  &\cdots	&m\in T		&\cdots \\
  \vdots	&\ddots &\vdots	&\adots	\\
  m'\in U	&\cdots	&[m\,m']&\cdots	\\
  \vdots	&\adots	&\vdots	&\ddots\\
}.\]
Whenever $U=\{1,x,\ldots,x^{k-1}\}$ and $T=\{1,x,\ldots,x^{\ell-1}\}$ then
$H_{U,T}$ is a classical Hankel matrix of size $k\times\ell$.

Definition~\ref{def:useful_staircase}
can be rewritten in term of a matrix rank.
\begin{definition}\label{def:useful_staircase_matrix}
  Let $T$ be a finite subset of terms.

  We say that a finite set
  $S\subset T$ is a \emph{useful staircase} with respect to $\bu$, $T$
  and $\prec$ if the matrix $H_{S,S}$
  has full rank equal to $\#\,S$ and to $\rank H_{T,T}$, $S$ is minimal for
  the inclusion and for $\prec$.

  We compare two ordered sets
  for $\prec$ by seeing them as tuples of their elements and then
  comparing them lexicographically.

  In other words, $S$ is the column rank profile of matrix $H_{T,T}$.
\end{definition}

As noted by the authors, it is important to notice that useful
staircases need not be \gbs staircases as proven by the following
example. Though, if the set of terms $T$ contains the true staircase
of the ideal of relations of $I$ with respect to $\prec$, then the
useful staircase will be this staircase, as expected.
\begin{example}\label{ex:useful_not_staircase}
  We consider the bivariate sequence
  $\bu=(\mathds{1}_{i=j=1})_{(i,j)\in\N^2}=\pare{
    \begin{smallmatrix}
      0 &0 &0 &0 &\cdots\\
      0 &1 &0 &0 &\cdots\\
      0 &0 &0 &0 &\cdots\\
      0 &0 &0 &0 &\cdots\\
      \vdots &\vdots &\vdots &\vdots &\ddots
    \end{smallmatrix}
  }$ whose ideal of relations is $\langle y^2,x^2\rangle$. The
  useful staircase with respect to $\bu$, $T=\{1,y,x,y^2\}$ and
  $\DRL(y\prec x)$ is $S=\{y,x\}$, as the columns labeled with
  $1$ and $y^2$ of the matrix
  \[H_{T,T}=
  \kbordermatrix{
    	&1	&y	&x	&y^2\\
    1	&0	&0	&0	&0\\
    y	&0	&0	&1	&0\\
    x	&0	&1	&0	&0\\
    y^2	&0	&0	&0	&0
  }\]
  are zero. However, for a bigger set $T'=\{1,y,x,y^2,x\,y,x^2\}$, the
  useful staircase of the matrix
  \[H_{T',T'}=
  \kbordermatrix{
    	&1	&y	&x	&y^2	&x\,y	&x^2\\
    1	&0	&0	&0	&0	&1	&0\\
    y	&0	&0	&1	&0	&0	&0\\
    x	&0	&1	&0	&0	&0	&0\\
    y^2	&0	&0	&0	&0	&0	&0\\
    x\,y&1	&0	&0	&0	&0	&0\\
    x^2	&0	&0	&0	&0	&0	&0
  }\]
  is the true staircase $\{1,y,x,x\,y\}$, which is stable by division.
\end{example}

\begin{proposition}
  If $S$ is the useful staircase with respect to the finite subset $T$ and
  $\prec$,
  then for all $m\in T\setminus S$, there exists a relation with
  support in $S\cup\{m\}$, but not in $S$, with a shift $T$.

  In particular, we can always pick $m$ in the border of $S$.
\end{proposition}
\begin{proof}
  If $m\in T\setminus S$, then $S\cup\{m\}$ is bigger than $S$.
  As the rank of $H_{T,S\cup\{m\}}$ cannot be
  $\#\,S\cup\{m\}=\#\,S+1=\rank H_{T,S}+1$, then it must be
  $\#\,S$. Therefore, the
  last column
  of $H_{T,S\cup\{m\}}$, labeled with $m$, is a linear combination of the
  previous ones, \ie there is a relation with support in
  $S\cup\{m\}$ but not in $S$.
\end{proof}
Finding this relation is straightforward, as it suffices to solve the
nondegenerate linear system $H_{T,S}\,\balpha+H_{T,\{m\}}=0$ which is
equivalent to solving $H_{S,S}\,\balpha+H_{S,\{m\}}=0$.

It is worth noticing that nothing can be concluded on the existence of
a relation with support in $S\cup\{m\}$ with a shift $T$, whenever
$m\not\in T$, though.

In the \sFGLM algorithm presented
in~\cite{issac2015,berthomieu:hal-01253934},
a relation was
returned for every $m$ in the border of $S$, whether $m$ was in $T$ or
not by solving the linear system
$H_{S,S}\,\balpha+H_{S,\cup\{m\}}=0$. This would mean
that some relations could be returned without even being tested with a shift
$T$, see also Example~\ref{ex:0dim}.
Therefore, it seems preferable to only return relations
with support in $T$, to ensure the shift $T$.

This yields the Algorithm~\ref{algo:sfglm}
that differs thus a little bit
from the one in the aforementioned articles.

\begin{algorithm2e}[htbp!]\label{algo:sfglm}
  \small
  \DontPrintSemicolon
  \TitleOfAlgo{The \sFGLM algorithm.}
  \KwIn{A table $\bu=(u_{\bi})_{\bi\in\N^n}$ with coefficients in
    $\K$, $\prec$ a monomial ordering and $T$ a set of terms in
    $\x$ stable by division.}
  \KwOut{A reduced truncated \gb with respect to $\prec$ of the ideal
    of relations of $\bu$ with staircase included in $T$.}
  Build the matrix $H_{T,T}$.\;
  Compute the useful staircase (column rank profile) $S$ of $H_{T,T}$
  such that $\rank H_{T,T}=\rank H_{S,S}$.\;
  $S':=\Stabilize(S)$.\tcp*{the staircase (stable under division)}
  $L:=T\backslash S'$.\tcp*{the set of next terms to study}
  $G:=\emptyset$.\tcp*{the future \gb}
  \While{$L\neq\emptyset$}{
    $t:=\min_{\prec}(L)$.\;
    Find $\balpha$ such that $H_{S,S}\,\balpha + H_{S,\{t\}}=0$.\;
    $G:=G\cup\left\{t+\sum_{s\in S}\alpha_s\,s\right\}$.\;
    Remove multiples of $t$ in $L$ and sort $L$ by increasing order
    (with respect to $\prec$).
  }
  \KwRet $G$.
\end{algorithm2e}

\begin{example}\label{ex:sfglm}
  We give the trace of the algorithm called on two sequences: the
  sequence $\bu=(2^i\,3^j\,(i+1))_{(i,j)\in\N^2}$ and the the binomial sequence
  $\bin$ with the $\DRL(y\prec x)$ ordering, and on the set
  $T=\{1,y,x,y^2,x\,y,x^2\}$ of
  all the monomials of degree at most $d=2$.

  \begin{enumerate}
  \item \begin{enumerate}
    \item[] We build the matrix
      \[H_{T,T}= \kbordermatrix{
        	&1	&y	&x	&y^2	&x\,y	&x^2\\
        1	&1	&3	&4	&9	&12	&12\\
        y	&3	&9	&12	&27	&36	&36\\
        x	&4	&12	&12	&36	&36	&32\\
        y^2	&9	&27	&36	&81	&108	&108\\
        x\,y	&12	&36	&36	&108	&108	&96\\
        x^2	&12	&36	&32	&108	&96	&80 }.
      \]
    \item[] The useful staircase of this matrix is
      $S=\{1,x\}$.
    \item[] It is stable by division so $S'=S$.
    \item[] We set
      $L=\{1,y,x,y^2,x\,y,x^2\}\setminus\{1,x\}
      =\{y,y^2,x\,y,x^2,y^3,x\,y^2,x^2\,y,x^3\}$ and $G=\emptyset$.
    \item[] We take $t=y$ and solve
      $H_{S,S}\,\balpha+H_{S,\{y\}}=0$ which yields relation
      $y-3$, so $G=\{y-3\}$ and $L$ is updated to
      $\{x^2,x^3\}$.
    \item[] We take $t=x^2$ and solve
      $H_{S,S}\,\balpha+H_{S,\{x^2\}}=0$ which yields relation
      $x^2-4\,x+4$, so $G=\{y-3,x^2-4\,x+4\}$ and $L$ is updated to
      $\emptyset$.
    \item[] We return $G=\{y-3,x^2-4\,x+4\}$.

    Furthermore, the relations $g\in G$ satisfy $[m\,g]=0$, for all
    $m\in T=\{1,y,x,y^2,x\,y,x^2\}$, \ie have a shift $T$.
    \end{enumerate}
  \item \begin{enumerate}
    \item[] We build the matrix
      \[H_{T,T}= \kbordermatrix{
        &1	&y	&x	&y^2	&x\,y	&x^2\\
        1	&1	&0	&1	&0	&1	&1\\
        y	&0	&0	&1	&0	&0	&2\\
        x	&1	&1	&1	&0	&2	&1\\
        y^2	&0	&0	&0	&0	&0	&1\\
        x\,y	&1	&0	&2	&0	&1	&3\\
        x^2	&1	&2	&1	&1	&3	&1 }.
      \]
    \item[] The useful staircase of this matrix is
      $S=\{1,y,x,y^2,x^2\}$.
    \item[] It is stable by division so $S'=S$.
    \item[] We set
      $L=\{1,y,x,y^2,x\,y,x^2\}\setminus\{1,y,x,y^2,x^2\}
      =\{x\,y,x\,y^2,x^2\,y\}$ and $G=\emptyset$.
    \item[] We take $t=x\,y$ and solve
      $H_{S,S}\,\balpha+H_{S,\{x\,y\}}=0$ which yields relation
      $x\,y-y-1$, so $G=\{x\,y-y-1\}$ and $L$ is updated to
      $\emptyset$.
    \item[] We return $G=\{x\,y-y-1\}$.

    Furthermore, this relation $g\in G$ satisfies $[m\,g]=0$, for all
    $m\in T=\{1,y,x,y^2,x\,y,x^2\}$, \ie has a shift $T$.
    \end{enumerate}
  \end{enumerate}
\end{example}


\section{Another linear algebra solver inspired by the \BMS
  algorithm}\label{s:New Solver}
In this section, we design an algorithm for computing the ideal of
relations of a sequence that is close to both the \BMS algorithm and
to the \sFGLM algorithm. The main idea will be to increase the number
of rows and columns of several multi-Hankel matrices and to check whether
the ranks of these matrices are increasing.

\begin{proposition}
  Let $S$ be a staircase and $g$ be a relation on sequence $\bu$ such
  that $\LM(g)$ lies on the border of $S$ and
  $\supp(g)\subseteq S\cup\{\LM(g)\}$.
  Assume furthermore that $g$ has a shift $m$, that is
  $[g\,\mu]=0$ for all $\mu\preceq m$.

  Let $m^+$ be the successor of $m$. If $[m^+\,g]\neq
  0$, then the linear system
  \[
  \begin{cases}
    \sum_{s\in S}\alpha_s\,[s]+[\LM(g)]&=0\\
    &\vdots\\
    \sum_{s\in S}\alpha_s\,[m\,s]+[m\,\LM(g)]&=0\\
    \sum_{s\in S}\alpha_s\,[m^+\,s]+[m^+\,\LM(g)]&=0\\
  \end{cases}
  \]
  has no solution and there is no nonzero valid relation with support
  in $\Stabilize\pare{S\cup\{m^+\}}$.
\end{proposition}
\begin{proof}
  This is a consequence of Proposition~\ref{prop:iter}.
\end{proof}

\begin{example}
  \begin{enumerate}
  \item Let us consider the binomial sequence $\bin$ and relations $y$
    and $x-1$. We know that the relation $y$ has a shift
    $y$, \ie $[y]=[y^2]=0$, and we want to check if a
    relation with leading monomial $y$ has a shift $x$. Therefore,
    we need to solve
    \[
    \begin{cases}
      \alpha_1\,[1]+[y]&=0\\
      \alpha_1\,[y]+[y^2]&=0\\
      \alpha_1\,[x]+[x\,y]&=0
    \end{cases}\iff
    \begin{cases}
      \alpha_1&=0\\
      0&=0\\
      \alpha_1+1&=0
    \end{cases}
    \]
    which has no solution. Hence $x$ is in the staircase. Thanks to
    Proposition~\ref{prop:iter}, since relation $[y]$ fails in $x\,y$
    for $[x\,y]=1$, we can also determine that $x$ is in the staircase.

    Likewise, we know that the relation $x-1$ has
    a shift $1$, \ie $[x-1]=0$, and we want to check if a
    relation with leading monomial $x$ has a shift $y$. Therefore,
    we need to solve
    \[
    \begin{cases}
      \alpha_1\,[1]+[x]&=0\\
      \alpha_1\,[y]+[x\,y]&=0
    \end{cases}\iff
    \begin{cases}
      \alpha_1+1&=0\\
      1&=0
    \end{cases}
    \]
    which has no solution. Hence $y$ is in the staircase.
  \item We still consider the binomial sequence $\bin$ but with
    relations $y^2,x\,y-1$ and $x^2-x$. We know that the relation
    $x^2-x$ has a shift $1$, \ie $[x^2-x]=0$, and we want to
    check if a 
    relation with leading monomial $x^2$ has a shift
    $y$. Therefore, we need to
    solve
    \[
    \begin{cases}
      \alpha_1\,[1]+\alpha_y\,[y]+\alpha_x\,[x]+[x^2]&=0\\
      \alpha_1\,[y]+\alpha_y\,[y^2]+\alpha_x\,[x\,y]+[x^2\,y]&=0
    \end{cases}\iff
    \begin{cases}
      \alpha_1+\alpha_x+1&=0\\
      \alpha_x+2&=0\\
    \end{cases}
    \]
    whose solution is $\alpha_x=-2,\alpha_1=1$ and $\alpha_y$ is
    any. Hence, although the relation $x^2-x$ fails at $x^2\,y$ for
    $[x^2\,y-x\,y]=1$, the relation $x^2-2\,x+1$ does not and has a
    shift $y$.
  \end{enumerate}
\end{example}

This yields Algorithm~\ref{algo:linalg_solver}.

\begin{algorithm2e}[htbp!]\label{algo:linalg_solver}
  \small
  \DontPrintSemicolon
  \TitleOfAlgo{Linear Algebra solver.}
  \KwIn{A table $\bu=(u_{\bi})_{\bi\in\N^n}$ with coefficients in
    $\K$, a monomial ordering $\prec$ and a monomial $M$
    as the stopping condition.}
  \KwOut{A set $G$ of relations generating $I_M$.}
  $T := \{m\in\K[\x],\ m\preceq M\}$.\tcp*{ordered for $\prec$}
  $G := \{[1,\emptyset]\}$.\tcp*{the future Gb, elements will be
    $[g,V_g]$}
  $S := \emptyset$.\tcp*{the staircase} 
  
  \Forall{$m\in T$}{
    $S' := S$ \;
    \For{$g\in G$}{
      \If{$\LM(g)| m$}{
        \uIf{$\rank H_{V_g\cup\acc{\frac{m}{\LM(g)}},S}<
          \rank H_{V_g\cup\acc{\frac{m}{\LM(g)}},S\cup\acc{\LM(g)}}$}{
          $S':=S'\cup\acc{\frac{m}{\LM(g)}}$.
        }\Else{
          $V_g:= V_g\cup\acc{\frac{m}{\LM(g)}}$.\;
        }
      }
    }
    $S:= \Stabilize(S')$. \;
    $G:= \Border (S)$.\;
    \For{$g\in G$}{
      $V_g:= \acc{\mu\in\K[\x],\mu\,\LM(g)\preceq m}$\;
    }  
  }
  \For{$g\in G$}{
    Find $\balpha$ such that $H_{V_g,S}\,\balpha + H_{V_g,\{\LM(g)\}}=0$.\;
    $g:=g+\sum_{s\in S}\alpha_s\,s$.\;
  }  
  \KwRet $G$.
\end{algorithm2e}

\begin{example}
  We detail how Algorithm~\ref{algo:linalg_solver} behaves on the
  binomial sequence $\bin$
  up to monomial $x^2$. We start with the empty staircase $S$ and
  the relation $1$, with $V_1=\emptyset$.
  \begin{enumerate}
  \item[] For the monomial $1$, the matrix
    $H_{\{1\},\emptyset}$ has rank $0$
    while the matrix
    $H_{\{1\},\{1\}}=\pare{\begin{smallmatrix}1\end{smallmatrix}}$ has rank
    $1$, hence $S$ is updated to $\{1\}$ and the relations are now
    $y$, with $V_y=\emptyset$, and $x$, with $V_x=\emptyset$.
  \item[] For the monomial $y$,
    \begin{enumerate}
    \item[] both matrices
      $H_{\{1\},\{1\}}=\pare{\begin{smallmatrix}1\end{smallmatrix}}$
      and
      $H_{\{1\},\{1,y\}}=\pare{\begin{smallmatrix}1
          &0\end{smallmatrix}}$
      have rank $1$, hence $V_y$ is updated to $\{1\}$;
    \item[] as $x$ does not divide $y$, nothing is done.
    \end{enumerate}
  \item[] For the monomial $x$,
    \begin{enumerate}
    \item[] as $y$ does not divide $x$, nothing is done;
      
    \item[] both matrices
      $H_{\{1\},\{1\}}=\pare{\begin{smallmatrix}1\end{smallmatrix}}$
      and
      $H_{\{1\},\{1,x\}}=\pare{\begin{smallmatrix}1
          &1\end{smallmatrix}}$
      have rank $1$, hence $V_x$ is updated to $\{1\}$.
    \end{enumerate}
  \item[] For the monomial $y^2$,
    \begin{enumerate}
    \item[] both matrices
      $H_{\{1,y\},\{1\}}=\pare{\begin{smallmatrix}
          1\\0\end{smallmatrix}}$
      and
      $H_{\{1,y\},\{1,y\}}=\pare{\begin{smallmatrix}
          1 &0\\
          0 &0\end{smallmatrix}}$
      have rank $1$, hence $V_y$ is updated to $\{1,y\}$;
    \item[] as $x$ does not divide $y$, nothing is done.
    \end{enumerate}
  \item[] For the monomial $x\,y$,
    \begin{enumerate}
    \item[] the matrix
      $H_{\{1,y,x\},\{1\}}=\pare{\begin{smallmatrix}
          1\\0\\1\end{smallmatrix}}$ has rank $1$ while the matrix
      $H_{\{1,y,x\},\{1,y\}}=\pare{\begin{smallmatrix}
          1 &0\\
          0 &0\\
          1 &1
        \end{smallmatrix}}$
      has rank $2$, hence $S$ is updated to $\{1,y\}$.
    \item[] the matrix
      $H_{\{1,y\},\{1\}}=\pare{\begin{smallmatrix}
          1\\0\end{smallmatrix}}$ has rank $1$ while the matrix
      $H_{\{1,y\},\{1,x\}}=\pare{\begin{smallmatrix}
          1 &0\\
          0 &1
        \end{smallmatrix}}$
      has rank $2$, hence $S$ is updated to $\{1,y,x\}$ and the
      relations are now 
      $y^2$, with $V_{y^2}=\{1\}$, $x\,y$, with $V_{x\,y}=\{1\}$,
      and $x^2$, with $V_{x^2}=\emptyset$.
    \end{enumerate}
  \item[] For the monomial $x^2$,
    \begin{enumerate}
    \item[] as $y^2$ does not divide $x^2$, nothing is done;
    \item[] as $x\,y$ does not divide $x^2$, nothing is done;
    \item[] both matrices
      $H_{\{1\},\{1,y,x\}}=\pare{\begin{smallmatrix}
          1 &0 &1\end{smallmatrix}}$
      and
      $H_{\{1\},\{1,y,x,x^2\}}=\pare{\begin{smallmatrix}
          1 &0 &1 &1\end{smallmatrix}}$
      have rank $1$, hence $V_{x^2}$ is updated to $\{1\}$.
    \end{enumerate}
  \item[] For the monomial $y^3$,
    \begin{enumerate}
    \item[] both matrices
      $H_{\{1,y\},\{1,y,x\}}=\pare{\begin{smallmatrix}
          1 &0 &1\\
          0 &0 &1\end{smallmatrix}}$
      and
      $H_{\{1,y\},\{1,y,x,y^2\}}=\pare{\begin{smallmatrix}
          1 &0 &1 &0\\
          0 &0 &1 &0\end{smallmatrix}}$
      have rank $2$, hence $V_{y^2}$ is updated to $\{1,y\}$.
    \item[] as $x\,y$ does not divide $y^3$, nothing is done;
    \item[] as $x^2$ does not divide $y^3$, nothing is done.
    \end{enumerate}
  \item[] For the monomial $x\,y^2$,
    \begin{enumerate}
    \item[] both matrices
      $H_{\{1,y,x\},\{1,y,x\}}=\pare{\begin{smallmatrix}
          1 &0 &1\\
          0 &0 &1\\
          1 &1 &1\end{smallmatrix}}$
      and
      $H_{\{1,y,x\},\{1,y,x,y^2\}}=\pare{\begin{smallmatrix}
          1 &0 &1 &0\\
          0 &0 &1 &0\\
          1 &1 &1 &0\end{smallmatrix}}$
      have rank $3$, hence $V_{y^2}$ is updated to $\{1,y,x\}$.
    \item[] both matrices
      $H_{\{1,y\},\{1,y,x\}}=\pare{\begin{smallmatrix}
          1 &0 &1\\
          0 &0 &1\end{smallmatrix}}$
      and
      $H_{\{1,y\},\{1,y,x,x\,y\}}=\pare{\begin{smallmatrix}
          1 &0 &1 &1\\
          0 &0 &1 &0\end{smallmatrix}}$
      have rank $2$, hence $V_{x\,y}$ is updated to $\{1,y\}$.
    \item[] as $x^2$ does not divide $x\,y^2$, nothing is done.
    \end{enumerate}
  \item[] For the monomial $x^2\,y$,
    \begin{enumerate}
    \item[] as $y^2$ does not divide $x^2\,y$, nothing is done;
    \item[] both matrices
      $H_{\{1,y,x\},\{1,y,x\}}=\pare{\begin{smallmatrix}
          1 &0 &1\\
          0 &0 &1\\
          1 &1 &1\end{smallmatrix}}$
      and
      $H_{\{1,y,x\},\{1,y,x,x\,y\}}=\pare{\begin{smallmatrix}
          1 &0 &1 &1\\
          0 &0 &1 &0\\
          1 &1 &1 &2\end{smallmatrix}}$
      have rank $2$, hence $V_{x\,y}$ is updated to $\{1,y,x\}$.
    \item[] both matrices
      $H_{\{1,y\},\{1,y,x\}}=\pare{\begin{smallmatrix}
          1 &0 &1\\
          0 &0 &1\end{smallmatrix}}$
      and
      $H_{\{1,y\},\{1,y,x,x^2\}}=\pare{\begin{smallmatrix}
          1 &0 &1 &1\\
          0 &0 &1 &2\end{smallmatrix}}$
      have rank $2$, hence $V_{x^2}$ is updated to $\{1,y\}$.
    \end{enumerate}
  \item[] For the monomial $x^3$,
    \begin{enumerate}
    \item[] as $y^2$ does not divide $x^3$, nothing is done;
    \item[] as $x\,y$ does not divide $x^3$, nothing is done;
    \item[] both matrices
      $H_{\{1,y,x\},\{1,y,x\}}=\pare{\begin{smallmatrix}
          1 &0 &1\\
          0 &0 &1\\
          1 &1 &1\end{smallmatrix}}$
      and
      $H_{\{1,y,x\},\{1,y,x,x^2\}}=\pare{\begin{smallmatrix}
          1 &0 &1 &1\\
          0 &0 &1 &2\\
          1 &1 &1 &1\end{smallmatrix}}$
      have rank $3$, hence $V_{x^2}$ is updated to $\{1,y,x\}$.
    \end{enumerate}
  \end{enumerate}
  Solving the linear systems yields relations $y^2$, with a shift $x$,
  $x\,y-y-1$, with a shift $x$, and $x^2-2\,x+1$, with a shift $x$.
\end{example}


\section{Analogies and differences}\label{s:comparison}
In this section, we present a list of similarities and differences of
behaviors and output for the \BMS and the \sFGLM algorithms. This
should convince the reader that these algorithms are not the same and
that it is not possible to tweak one of them to mimic the behavior
of the other.

\subsection{Closed staircase}
Although both algorithms compute first a set of elements in the
staircase, one of the main differences between the \BMS and the \sFGLM
algorithms
is how they handle the leading terms outside of this staircase.

\begin{theorem}\label{th:closed_staircase}
  Let $\bu$ be a sequence and $\prec$ be a monomial ordering.
  
  Calling the \BMS algorithm  on
  $\bu$, $\prec$ and a stopping monomial $M$ yields a
  truncated \gb of a zero-dimensional ideal.

  Calling the \sFGLM algorithms on $\bu$, $\prec$ and a set of
  terms $T$ stable by division yields a truncated \gb of an
  ideal, with leading monomials in $T$,
  which is not necessarily a zero-dimensional ideal.

  Furthermore, if the \BMS and the
  \sFGLM algorithms compute the ideal of relations of $\bu$, then the
  ideal computed by the \BMS algorithm is included in the ideal
  computed by the \sFGLM algorithm. These ideals are equal if, and
  only if, $\bu$ is linear recurrent.
\end{theorem}
\begin{proof}
  The proof of the first part comes directly from the line
  $G':=\Border(S')$ in the description of the \BMS algorithm and then
  to the manipulations done to $g'\in G'$.

  The proof of the second part comes from the fact that the potential
  leading terms in the \sFGLM algorithm
  are taken in the intersection of the border of the
  staircase and the input set of terms. Nothing may ensure that this
  set has a pure power of every variable. See also
  Example~\ref{ex:0dim}.
\end{proof}
This is
illustrated in the following examples.
\begin{example}\label{ex:0dim}
  \begin{enumerate}
  \item We let
    $\bu=\pare{i^2+j+\mathds{1}_{3\,i+2\,j > 9}}_{(i,j)\in\N^2}$ be a
    sequence and consider the $\DRL(y\prec x)$ ordering.

    The \BMS algorithm called on $\bu$ and the stopping monomial $y^3$
    returns the ideal of relations $\langle x-y,y^2-2\,y\rangle$.

    The \sFGLM algorithm called on $\bu$ and the set of terms
    $T=\{1,y,x,y^2\}$ returns the ideal of relations
    $\langle y^2-2\,y+1\rangle$.
  \item We consider now the binomial sequence $\bin$  and the
    $\DRL(y\prec x)$ ordering.

    The \BMS algorithm called on $\bin$ and the stopping monomial
    $x^5$ returns $\langle x\,y-y-1, y^3,(x-1)^3\rangle$.

    The \sFGLM algorithm called on $\bin$ and the set of terms
    $T$ of all the monomials of degree at most $3$
    returns $\langle x\,y-y-1\rangle$.
    
    The first ideal is obviously included in the second which is the
    true ideal of relations of the binomial sequence.
  \end{enumerate}
\end{example}

\begin{remark}
  It is possible to tweak the \sFGLM algorithm so that it \emph{tries}
  to close the staircase. The idea is to pick the potential leading
  terms in the border of the staircase. Then, for $t$ such a potential
  leading term, if $t$ is not in the input set of terms $T$, one tries
  to solve $H_{T,S}\,\balpha+H_{T,\{t\}}=0$ instead of only
  $H_{S,S}\,\balpha+H_{S,\{t\}}=0$, so that relation
  $t+\sum_{s\in S}\alpha_s\,s$ has a shift $T$.
  See Algorithm~\ref{algo:sfglm_0dim}.
\end{remark}
\begin{algorithm2e}[htbp!]\label{algo:sfglm_0dim}
  \small
  \DontPrintSemicolon
  \TitleOfAlgo{Tweaked \sFGLM algorithm.}
  \KwIn{A table $\bu=(u_{\bi})_{\bi\in\N^n}$ with coefficients in
    $\K$, $\prec$ a monomial ordering and $T$ a set of terms in
    $\x$ stable by division.}
  \KwOut{A reduced truncated \gb with respect to $\prec$ of the ideal
    of relations of $\bu$ with staircase included in $T$.}
  Build the matrix $H_{T,T}$.\;
  Compute the useful staircase (column rank profile) $S$ of $H_{T,T}$
  such that $\rank H_{T,T}=\rank H_{S,S}$.\;
  $S':=\Stabilize(S)$.\;
  $L:=\pare{T\cup\bigcup_{i=1}^n x_i\,S'}\backslash S'$.\;
  $G:=\emptyset$.\;
  \While{$L\neq\emptyset$}{
    $t:=\min_{\prec}(L)$.\;
    Find $\balpha$ such that $H_{S,S}\,\balpha + H_{S,\{t\}}=0$.\;
    \If(\tcp*[f]{has a shift $T$!}){$t\in T$ or
      $H_{T\setminus S,S}\,\balpha + H_{T\setminus S,\{t\}}=0$}{
      $G:=G\cup\left\{t+\sum_{s\in S}\alpha_s\,s\right\}$.\;
    }
    Remove multiples of $t$ in $L$ and sort $L$ by increasing order
    (with respect to $\prec$).
  }
  \KwRet $G$.
\end{algorithm2e}

Let us notice that this tweaked version of the \sFGLM still can fail
to close the staircase.
\begin{example}
  We call Algorithm~\ref{algo:sfglm_0dim} on sequence
  $\bu=\pare{i^2+j+\mathds{1}_{3\,i+2\,j > 9}}_{(i,j)\in\N^2}$,
  the set $T=\{1,y,x,y^2\}$ and the $\DRL(y\prec x)$ ordering as in
  Example~\ref{ex:0dim}.
  \begin{enumerate}
  \item[] We build the matrix
    \[H_{T,T} = \kbordermatrix{
      		&1	&y	&x	&y^2\\
      1		&0	&1	&1	&2\\
      y		&1	&2	&2	&3\\
      x		&1	&2	&4	&3\\
      y^2	&2	&3	&3	&4
    }.\]
  \item[] The useful staircase of this matrix is $S=\{1,y,x\}$.
  \item[] It is stable by division so $S'=S$.
  \item[] We set
    $L=\{1,y,x,y^2,x\,y,x^2\}\setminus\{1,y,x\}=\{y^2,x\,y,x^2\}$
    and $G=\emptyset$.
  \item[] We take $t=y^2$ and solve
    $H_{S,S}\,\balpha+H_{S,\{y^2\}}=0$
    which yields relation $y^2-2\,y+1$, so $G=\{y^2-2\,y+1\}$ and $L$ is
    updated to $\{x\,y,x^2\}$.
  \item[] We take $t=x\,y$ and solve
        \[H_{S,S\cup\{x\,y\}}\,
    \begin{pmatrix}
      \alpha_1\\\alpha_y\\\alpha_x\\1
    \end{pmatrix}
    = \kbordermatrix{
      &1	&y	&x	&x\,y\\
      1	&0	&1	&1	&2\\
      y	&1	&2	&2	&3\\
      x	&1	&2	&4	&5\\
    }\,
    \begin{pmatrix}
      \alpha_1\\\alpha_y\\\alpha_x\\1
    \end{pmatrix}
    =0,\]
    which yields
    relation $x\,y-x-y+1$. We check that
    \[H_{T\setminus S,S\cup\{x\,y\}}\,
    \begin{pmatrix}
      1\\-1\\-1\\1
    \end{pmatrix}
    = \kbordermatrix{
      		&1	&y	&x	&x^2\\
      y^2	&2	&3	&3	&4
    }\,
    \begin{pmatrix}
      1\\-1\\-1\\1
    \end{pmatrix}
    =0,\]
    set
    $G=\{y^2-2\,y+1,x\,y-x-y+1\}$ and update $L$ to
    $\{x^2\}$.
  \item[] We take $t=x^2$ and solve
    \[H_{S,S\cup\{x\,y\}}\,
    \begin{pmatrix}
      \alpha_1\\\alpha_y\\\alpha_x\\1
    \end{pmatrix}
    = \kbordermatrix{
      &1	&y	&x	&x\,y\\
      1	&0	&1	&1	&2\\
      y	&1	&2	&2	&3\\
      x	&1	&2	&4	&5\\
    }\,
    \begin{pmatrix}
      \alpha_1\\\alpha_y\\\alpha_x\\1
    \end{pmatrix}
    =0,\]
    which yields
    relation $x^2-2\,x^2-2\,x+3$. However,
    \[H_{T\setminus S,S\cup\{x^2\}}\,
    \begin{pmatrix}
      3\\-2\\-2\\1
    \end{pmatrix}
    = \kbordermatrix{
      		&1	&y	&x	&x^2\\
      y^2	&2	&3	&3	&7
    }\,
    \begin{pmatrix}
      3\\-2\\-2\\1
    \end{pmatrix}
    =1,\]
    so the relation is not valid when shifted by $y^2$.
    Hence, we let $G=\{y^2-2\,y+1,x\,y-x-y+1\}$ and update $L$ to
    $\emptyset$.
  \item[] We return $G=\{y^2-2\,y+1,x\,y-x-y+1\}$.
  \item[] Furthermore, each relation $g\in G$ satisfies $[m\,g]=0$, for all
    $m\in T=\{1,y,x,y^2\}$, \ie has a shift $T$.
  \end{enumerate}
\end{example}

\subsection{Reduction of relations}
Even though the \BMS and the \sFGLM algorithms may compute the same
ideal of relations for a given sequence, the \gbs they compute  may
differ. However, it is possible to tweak the \BMS algorithm so that it
returns the same \gb of the ideal as the \sFGLM algorithm.

\begin{theorem}\label{th:reduced_gb}
  Let $\bu$ be a sequence and $\prec$ be a monomial ordering.
  
  Calling the \sFGLM algorithms on $\bu$, $\prec$ and a set of
  terms $T$ stable by division yields a truncated minimal
  reduced \gb of an ideal.

  Calling the \BMS algorithm  on
  $\bu$, $\prec$ and a stopping monomial $M$ yields a
  truncated minimal \gb of an ideal, which is not necessarily reduced.

  Furthermore, even if $\bu$ is linear recurrent and the
  \sFGLM algorithm computes the ideal of relations of $\bu$, then there
  is no reason for the output of the \BMS algorithm to be reduced.
\end{theorem}
\begin{proof}
  When updating a relation $g$ thanks to a failing relation $h$ in the
  \BMS algorithm, nothing ensures that $g$ has support in
  $S\cup\{\LM(g)\}$, where $S$ is the current staircase, as $\supp h$
  may not be included in $S$. This prevents the returned \gb to be
  reduced, see also Example~\ref{ex:reduced}.

  As the \sFGLM algorithm computes a staircase $S$, the monomials on the
  border of $S$ and then solves a multi-Hankel linear system indexed by
  $S$ and one of the monomial on this border, it is clear that the
  output truncated \gb is reduced.
\end{proof}

The following example show which \gbs are returned by the \BMS and the \sFGLM
algorithms for a same sequence.
\begin{example}\label{ex:reduced}
  We let $\bu=\pare{i^2+j^2-1}_{(i,j)\in\N^2}$
  be a sequence and consider the $\DRL(y\prec x)$ ordering.
  The ideal of relations of $\bu$ is $I=\langle
  x\,y-x-y+1,x^2-y^2-2\,x+2\,y,y^3-3\,y^2+3\,y-1\rangle$.

  The \BMS algorithm called on $\bu$ and the stopping monomial $y^5$ returns
  $g_1=x\,y-x-y+1$, with shift $x^2$,
  $g_2=x^2-\frac{1}{3}\,x\,y-y^2-\frac{5}{3}\,x+\frac{7}{3}\,y-\frac{1}{3}$,
  with shift $x^2$ and
  $g_3=y^3-\frac{1}{2}\,x\,y-3\,y^2+\frac{1}{2}\,x+\frac{7}{2}\,y-\frac{3}{2}$,
  with shift $y^2$. We can notice that $\{g_1,g_2,g_3\}$ is a minimal \gb but
  not a reduced \gb of $I$.
  
  The \sFGLM algorithm called on $\bu$ and the set of all the
  monomials of degree at most $3$ yields relations
  $g_1'=x\,y-x-y+1,g_2'=x^2-y^2-2\,x+2\,y,g_3'=y^3-3\,y^2+3\,y-1$.
  We can notice that $\{g_1',g_2',g_3'\}=\{g_1,g_2+\frac{1}{3}\,g_1,g_3+\frac{1}{2}\,g_1\}$
  is a minimal reduced \gb of $I$.
\end{example}

\begin{remark}
  Let $g_1',g_2'$ be two computed relations by the \BMS algorithm and
  let $\mu$ be a monomial. Assume $\mu\,\LM(g_1')\preceq \LM(g_2')$,
  then $\shift(\mu\,\LM(g_1'))\succeq \LM(g_2')=v$. Therefore
  $g_2'-c\,\mu\,g_1'$ has still shift $v$ for any scalar $c$: hence
  one can replace $g_2'$ by $g_2'-c\,\mu\,g_1'$, \ie one can reduce
  $g_2'$ by $g_1'$ into $g_2$ and replace $g_2'$ by $g_2$.
  Let us notice that we can tweak the \BMS algorithm so that, at each
  step, the set of relations is a truncated reduced \gb. It suffices to
  perform an inter-reductions of the computed relations at the end of
  each step of the \textbf{For} loop, see
  Algorithm~\ref{algo:tweaked_bms}.
\end{remark}

\begin{algorithm2e}[htbp!]\label{algo:tweaked_bms}
  \small
  \DontPrintSemicolon
  \TitleOfAlgo{Tweaked \BMS algorithm}
  \KwIn{A table $\bu=(u_{\bi})_{\bi\in\N^n}$ with coefficients in
    $\K$, a monomial ordering $\prec$ and a monomial $M$
    as the stopping condition.}
  \KwOut{A set $G$ of relations generating $I_M$.}
  $T := \{m\in\K[\x],\ m\preceq M\}$.\;
  $G := \{1\}$.\;
  $S := \emptyset$.\;
  \Forall{$m\in T$}{
    $S' := S$ \;
    \For{$g\in G$}{
      \If{$\LM(g)| m$}{
        $e:=\cro{\frac{m}{\LM(g)}\,g}_{\bu}$\;
        \If{$e\neq 0$}{
          $S':=S'\cup\acc{\cro{\frac{g}{e},\frac{m}{\LM(g)}}}$.\;
        }
      }
      $S':=\min_{\fail(h)\in S'}\acc{[h,\fail(h)/\LM(h)]}$.\;
      $G':= \Border (S')$.\;
    }
    \For{$g'\in G'$}{
      Let $g\in G$ such that $\LM(g)|\LM(g')$.\; 
      \uIf{$\LM(g) \nmid m$}{
        $g':=\frac{\LM(g')}{\LM(g)}\,g$.\;
      }
      \uElseIf{$\exists\,h\in S,
        \frac{m}{\LM(g')} |\fail(h)$}{
        $g':= \frac{\LM(g')}{\LM(g)}\,g
        -\cro{\frac{m}{\LM(h)}\,h}_{\bu}\,
        \frac{\LM(g')\,\fail(h)}{m}\,h$.\;
      }
      \lElse{
        $g':=g$.
      }
    }
    $G := \InterReduce(G')$ \;
    $S := S'$ \;
  }  
  \KwRet $G$.
\end{algorithm2e}

\subsection{Validity of relations}
We compare the relationship between relations and shifts as they are
computed by the \BMS and the \sFGLM algorithms.
\begin{theorem}\label{th:valid_shift}
  Let $\bu$ be a sequence and $\prec$ be a monomial ordering.

  Calling the \BMS algorithm on $\bu$, $\prec$ and a stopping monomial
  $M$ yields
  relations $g_1,\ldots,g_r$ and shifts $v_1,\ldots,v_r$ such that
  \[\forall\,i, 1\leq i\leq r,\quad v_i\,\LM(g_i)\preceq M\]
  and $g_i$ is valid with a shift $v_i$, potentially $0$.

  Calling the \sFGLM algorithm on $\bu$, $\prec$ and a set of terms
  $T_M=\{m,\ m\preceq M\}$
  yields relations $g_1',\ldots,g_{r'}'$ such that
  \[\forall\,i, 1\leq i\leq r',\quad \LM(g_i')\preceq M\]
  and $g_i'$ has a shift $T_M$, \ie is valid with a shift $M$.  
\end{theorem}
\begin{proof}
  The \BMS algorithm tests its relations up to $M$, \ie it shifts
  them up to $M$. In the worst
  case, the leading term of a relation is greater than $M$, but then
  it has a shift $0$.

  The \sFGLM algorithm returns relations
  $g=\LT(g)+\sum_{s\in S}\alpha_s\,s$
  such that $H_{T,S}\,\balpha+H_{T,\{\LM(g)\}}=0$, \ie they are valid
  when shifted by any monomial in $T$.
\end{proof}
We illustrate this with the following example.
\begin{example}\label{ex:validity}
  We let $\bin$ be the binomial sequence and consider the
  $\DRL(y\prec x)$ ordering.

  The \BMS algorithm called on $\bin$ and the stopping monomial $x^7$ returns
  $x\,y-y-1$, with a shift $x^5$; $y^4$, with a shift
  $x^3$; and $(x-1)^4$, with a shift $x^3$.
  
  With the matrix viewpoint, one has
  \begin{align*}
    \text{$21$ rows}\left\{\kbordermatrix{
    	&1	&y	&x	&x\,y\\
    1	&1	&0	&1	&1\\
    y	&0	&0	&1	&0\\
    x	&1	&1	&1	&2\\
    \vdots	&\vdots	&\vdots	&\vdots	&\vdots\\
    x^3	&1	&3	&1	&4\\
    \vdots	&\vdots	&\vdots	&\vdots	&\vdots\\
    x^5	&1	&5	&1	&6\\
    }\right.\,
    \begin{pmatrix}
      -1\\-1\\0\\1
    \end{pmatrix}
    &=0,\\
    \text{$10$ rows}\left\{\kbordermatrix{
    &1	&y	&x	&\cdots	&y^4\\
    1	&1	&0	&1	&\cdots	&0\\
    y	&0	&0	&1	&\cdots	&0\\
    x	&1	&1	&1	&\cdots	&0\\
    \vdots	&\vdots	&\vdots	&\vdots	&	&\vdots\\
    x^3	&1	&3	&1	&\cdots	&0\\
    }\right.\,
    \begin{pmatrix}
      0\\0\\0\\\vdots\\1
    \end{pmatrix}
    &=0,
    &\kbordermatrix{
    &1	&y	&x	&\cdots	&x^4\\
    1	&1	&0	&1	&\cdots	&1\\
    y	&0	&0	&1	&\cdots	&4\\
    x	&1	&1	&1	&\cdots	&1\\
    \vdots	&\vdots	&\vdots	&\vdots	&	&\vdots\\
    x^3	&1	&1	&1	&\cdots	&1\\
    }\,
    \begin{pmatrix}
      1\\0\\-4\\\vdots\\1
    \end{pmatrix}
    &=0.
  \end{align*}
  We can notice that the first matrix has many more rows than the other two.

  The \sFGLM algorithm called on $\bin$ and the set $T$ of all the
  monomials of degree at most $3$ returns $x\,y-y-1$ with a shift $x^3$.
  With the matrix viewpoint, one has
  \[
  \text{$10$ rows}\left\{\kbordermatrix{
    	&1	&y	&x	&x\,y\\
    1	&1	&0	&1	&1\\
    y	&0	&0	&1	&0\\
    x	&1	&1	&1	&2\\
    \vdots	&\vdots	&\vdots	&\vdots	&\vdots\\
    x^3	&1	&3	&1	&4\\
  }\right.\,
  \begin{pmatrix}
    -1\\-1\\0\\1
  \end{pmatrix}
  =0.\]

  Likewise, calling Algorithm~\ref{algo:sfglm_0dim} on the same input
  returns $x\,y-y-1,y^4,x^4$, all three valid up to $x^3$. With the
  matrix viewpoint, we also have this matrix equality:
  \[
  \text{$10$ rows}\left\{\kbordermatrix{
    	&1	&y	&x	&\cdots	&y^4\\
    1	&1	&0	&1	&\cdots	&0\\
    y	&0	&0	&1	&\cdots	&0\\
    x	&1	&1	&1	&\cdots	&0\\
    \vdots	&\vdots	&\vdots	&\vdots	&	&\vdots\\
    x^3	&1	&3	&1	&\cdots	&0
  }\right.\,
  \begin{pmatrix}
    0\\0\\0\\\vdots\\1
  \end{pmatrix}
  =\kbordermatrix{
    	&1	&y	&x	&\cdots	&x^4\\
    1	&1	&0	&1	&\cdots	&1\\
    y	&0	&0	&1	&\cdots	&4\\
    x	&1	&1	&1	&\cdots	&1\\
    \vdots	&\vdots	&\vdots	&\vdots	&	&\vdots\\
    x^3	&1	&3	&1	&\cdots	&1
  }\,
  \begin{pmatrix}
    1\\0\\-4\\\vdots\\1
  \end{pmatrix}
  =0.\]
  We can see that the last two matrices have as many rows as the first one.
\end{example}

That being said, for a monomial $m=\mu\,\LM(g)\in T$, the column
labeled with $m$ in $H_{T,T}$ is also linearly dependent from the
previous ones. In particular, it allows us to verify that the relation
$\mu\,g$ is valid with a shift $T$, \ie $g$ is valid with a shift
$T\cup\mu\,T$.

\begin{example}
  Resuming Example~\ref{ex:validity}, the columns labeled with
  $x\,y^2,x^2\,y,x\,y^3,x^2\,y^2$ and $x^3\,y$ are linearly dependent from
  the previous ones. For instance, the column labeled with $x\,y^2$ is
  the sum of the columns labeled with $y^2$ and $y$. Thus, Pascal's
  rule $x\,y-y-1$ is also valid with shifts $y\,T$, $x\,T$, $y^2\,T$,
  $x\,y\,T$ and $x^2\,T$. Since $T$ is the set of the monomials of
  degree at most $3$, $\bigcup_{\mu\in\{1,y,x,y^2,x\,y,x^2\}}\mu\,T$ is
  the set of all the monomials of degree $5$.

  All in all, like for the \BMS algorithm, we find that the Pascal's
  rule is valid with a shift $x^5$.  
\end{example}

\subsection{Monomial ordering and Set of Terms}
Given a linear recurrent sequence $\bu$ with ideal of relation defined by a
\gb $\cG$ for a monomial ordering $\prec$, the \BMS and the \sFGLM
algorithms can return $\cG$ only if the input set of terms
contains the staircase defined by $\cG$. That is why, it is preferable
to run both of them with an ordering $\prec$ such that for all
monomial $M\in\K[\x]$, $T_M=\{m,\ m\preceq M\}$ is finite. In
particular, the $\LEX$ ordering does not satisfy such a property.

However, we can try to see how they behave when calling them with the
$\LEX$ monomial ordering. We relate this to the randomized
reduction to the \BM algorithm presented
in~\cite[Section~3]{issac2015,berthomieu:hal-01253934} where the
authors perform a randomized  linear change of variables so that,
generically, the ideal of relations is in shape position. We also
relate this to the \spFGLM
algorithm application where the input is a
sequence made from a \gb, typically for the $\DRL$ ordering, and the
output is the ideal of relations of this sequence for another
ordering, typically $\LEX$, see~\cite{FM11,faugere:hal-00807540}.

\begin{theorem}\label{th:shape_position}
  Let $\bu$ be a linear recurrent sequence whose
  ideal of relation $I$ is in shape position for the
  $\LEX(x_n\prec\cdots\prec x_2\prec x_1)$ ordering, \ie there exist
  $g_n$ squarefree and
  $f_{n-1},\ldots,f_1\in\K[x_n]$ with $\deg g_n=d,\deg f_i<d$ such that
  $I=\langle g_n(x_n),x_{n-1}-f_{n-1}(x_n),\ldots,x_1-f_1(x_n)\rangle$.
  
  Calling the \sFGLM algorithm as designed
  in~\cite{issac2015,berthomieu:hal-01253934} or its tweaked version
  Algorithm~\ref{algo:sfglm_0dim} on $\bu$, $T$ containing at least
  $\{1,x_n,\ldots,x_n^{d-1}\}$ and
  $\LEX(x_n\prec\cdots\prec x_2\prec x_1)$ allows one to retrieve $I$.

  Calling the \BMS algorithm on $\bu$, with the stopping monomial
  $x_n^e$ and $\LEX(x_n\prec\cdots\prec x_2\prec x_1)$ yields $\langle
  g(x_n),x_{n-1},\ldots,x_1\rangle$. This ideal is not equal to $I$,
  unless $f_1=\cdots=f_{n-1}=0$.
  
  In other words, the \sFGLM algorithm can retrieve an ideal of
  relations in shape position while, in general, the \BMS algorithm cannot.
\end{theorem}
\begin{proof}
  When calling the \sFGLM algorithm on $\bu$ with the
  $\LEX (x_n\prec\cdots\prec x_2\prec x_1)$ ordering
  and with the set of terms
  $T$ containing $\{1,x_n,\ldots,x_n^{d-1}\}$, the algorithm shall
  determine that the useful staircase
  $S=\{1,x_n,\ldots,x_n^{d-1}\}$. Then, the set of potential leading
  monomials is $\{x_n^d,x_{n-1},\ldots,x_1\}$. For $x_n^d$, it solves
  $H_{S,S}\,\balpha+H_{S,\{x_n^d\}}=0$ and finds relation $g_n(x_n)$ while
  for any $k$, $1\leq k\leq n-1$, it solves
  $H_{S,S}\,\balpha+H_{S,\{x_k\}}=0$ and finds relation
  $x_k-f_k(x_n)$. Then it tests that these relations have a shift
  $T$, and since they are the true relations of $\bu$, they do.

  When calling the \BMS algorithm with the $\LEX(x_n\prec\cdots\prec
  x_2\prec x_1)$ ordering and with the stopping monomial $M=x_n^e$, the
  algorithm behaves \emph{mutatis mutandis} like the \BM algorithm
  except that as soon as $1$ is detected to be in the staircase, the
  \BMS algorithms adds polynomials $x_1,\ldots,x_{n-1}$ in the truncated
  \gb. As these relations cannot be tested further, the output shall
  always be $\cG=\langle g(x_n),x_{n-1},\ldots,x_1\rangle$. See also
  Example~\ref{ex:shape_position} below.
\end{proof}

We illustrate the behavior of the \sFGLM algorithm with an example.
\begin{example}\label{ex:shape_position}
  We let $\bu=\pare{F_{4\,i+k}}_{(i,j,k)\in\N^3}$
  be a sequence, where $(F_i)_{i\in\N}$ is the Fibonacci sequence,
  and consider the $\DRL(z\prec y\prec x)$ ordering.
  The ideal of relations of $\bu$ is $I=\langle
  z^2-z-1,y-1,x-3\,z-2\rangle$.

  The \sFGLM algorithm called on $\bu$ and the set of terms
  $\{1,z,\ldots,z^{d+2}\}$ yields $\langle g_3,g_2,g_1\rangle$, which
  is indeed the ideal of relations of the sequence. In detail:
  \begin{enumerate}
  \item[] It creates the matrix
    \[H_{T,T}= \kbordermatrix{
      &1	&z	&\cdots	&z^{d+2}\\
      1	&0	&1	&\cdots	&F_{d+2}\\
      z	&1	&1	&\cdots	&F_{d+3}\\
      \vdots	&\vdots	&\vdots	&	&\vdots\\
      z^{d+2} &F_{d+2}&F_{d+3}&\cdots&F_{2\,d+4} }\]
    and finds it has rank $2$ with useful staircase
    $S=\{1,z\}$.
  \item[] It solves
    \[H_{S,S\cup\{z^2\}}\,
    \begin{pmatrix}
      \alpha_1\\\alpha_{z}\\1
    \end{pmatrix}
    = \kbordermatrix{
      &1	&z	&z^2\\
      1	&0	&1	&1\\
      z &1 &1 &2 }\,
    \begin{pmatrix}
      \alpha_1\\\alpha_{z}\\1
    \end{pmatrix}
    =0\] and finds relation $z^2-z-1$.
  \item[] Then, it solves
    \[
      H_{S,S\cup\{y\}}\,
      \begin{pmatrix}
        \alpha_1\\\alpha_{z}\\1
      \end{pmatrix}
   =
     \kbordermatrix{
   &1	&z	&y\\
      1	&0	&1	&1\\
      z	&1	&1	&1\\
      }\,
      \begin{pmatrix}
        \alpha_1\\\alpha_{z}\\1
      \end{pmatrix}=
      H_{S,S\cup\{x\}}\,
      \begin{pmatrix}
        \alpha_1\\\alpha_{z}\\1
      \end{pmatrix}
   =
     \kbordermatrix{
   &1	&z	&x\\
      1	&0	&1	&3\\
      z	&1	&1	&5\\
      }\,
      \begin{pmatrix}
        \alpha_1\\\alpha_{z}\\1
      \end{pmatrix}=0.
    \]
    and finds the relations $g_2=y-1$ and $g_1=x-3\,z-2$. It also
    checks that the last two have a shift $T$ with
    \begin{align*}
      H_{T\setminus S,S\cup\{y\}}\,
      \begin{pmatrix}
        -1\\0\\1
      \end{pmatrix}
   &=
     \kbordermatrix{
   &1	&z	&y\\
      z^2	&1	&2	&1\\
      \vdots	&\vdots	&\vdots	&\vdots\\
      z^{d+2}	&F_{d+2}	&F_{d+3}	&F_{d+2}\\
      }\,
      \begin{pmatrix}
        -1\\0\\1
      \end{pmatrix}=0,\\
      H_{T\setminus S,S\cup\{x\}}\,
      \begin{pmatrix}
        -2\\-3\\1
      \end{pmatrix}&= \kbordermatrix{
   &1	&z	&x\\
      z^2	&1	&2	&8\\
      \vdots	&\vdots	&\vdots	&\vdots\\
      z^{d+2}	&F_{d+2}	&F_{d+3}	&F_{d+6}\\
      }\,
      \begin{pmatrix}
        -2\\-3\\1
      \end{pmatrix}=0.
    \end{align*}
  \item[] Finally, it returns $g_3,g_2,g_1$.
  \end{enumerate}

  The \BMS algorithm called on $\bu$ and the stopping monomial $z^{d+2}$
  returns $\langle g_3,y,x\rangle$, which
  is not the ideal of relations of the sequence, as neither $y$ nor
  $x$ are in $I$. In detail:
  
  \begin{enumerate}
  \item[] The algorithms tests the relation $g=1$ in
    $u_{0,0,0}=F_0=0$ where it succeeds.
  \item[] It tests $g$ in
    $u_{0,0,1}=F_1=1$ where it fails. It has now relations
    $g_1=x,g_2=y$ and $g_3=z^2$, all three with a shift $0$.
  \item[] Going on testing $z^2$ in $u_{0,0,2}=F_2=1$, $u_{0,0,3}=F_3$ and so
    on, it is able to update $g_3$ into $z^2-z-1$ but is never
    able to test either $g_1$ or $g_2$.
  \item[] Finally, it returns $g_3$ with a shift $z^d$ and $g_1,g_2$ with
    a shift $0$.
  \end{enumerate}
  Although $g_3$ is in the ideal of relations, $g_1$ and $g_2$ are not.
\end{example}
\begin{remark}
  Let us notice though that, whenever the user knows the degree $d$ of
  the ideal of relations of a linear recurrent sequence, we can tweak
  both algorithms to be able to recover fully the ideal of relations.

  On the one hand, it suffices to call the \sFGLM algorithm with the
  set of monomials $T=\{m,\ \deg m\leq d\}$ and the
  $\LEX(x_n\prec\cdots\prec x_2\prec x_1)$ ordering.

  On the other hand, it suffices to change a little bit how we
  enumerate the monomials less than the stopping monomial $M$ in the
  \BMS algorithm.  In most implementation, monomials less than or
  equal to $M$ are
  given by the ordered set of terms $T_M=\{m,\ m\preceq M\}$. If one knows
  in advance the degree $d$ of the ideal, then it suffices to
  enumerate the monomials in $\{m,\ \deg m\leq 2\,d-1,m\preceq M\}$ and to call
  the \BMS algorithm with the $\LEX(x_n\prec\cdots\prec x_2\prec x_1)$
  ordering. This tweaked version of the \BMS algorithm was implemented
  for the \spFGLM application in~\cite{FM11,faugere:hal-00807540}.
\end{remark}


\section{Complexity and Benchmarks}\label{s:implem}
In this section, we present some benchmarks to compare the behaviors
of the \BMS and the \sFGLM algorithms.
We relate them with the announced
complexity of each algorithm.

Three families of ideals of relations are used to make the
sequences.

\begin{itemize}
\item In the first family, 
  the leading monomials of the ideal of
  relations are $\langle y^{\floor{d/2}},x^d\rangle$. Thus, its staircase is a
  rectangle of size around $d^2/2$. In three variables,
  the leading monomials are
  $\langle z^{\ceil{d/3}},y^{\floor{d/2}},x^d\rangle$,
  so that the staircase is a rectangular cuboid of size
  around $d^3/6$. This family will be called \emph{Rectangle}.

\item In the second family,
  the leading monomials of the ideal of relations
  are $\langle x\,y,y^d,x^d\rangle$.  Thus, its staircase looks like a
  \textsc{L} and has size $2\,d-1$.
  In three variables, the leading monomials are
  $\langle y\,z,x\,z,x\,y,z^d,y^d,x^d\rangle$, so that the staircase has size
  $3\,d-2$. This family will be called \emph{\textsc{L}
    shape}. It was considered as the worst case
  in~\cite{issac2015,berthomieu:hal-01253934} for the \asFGLM algorithm, a
  variant of the \sFGLM algorithm, for the number of queries.
  It should also be a worst case for
  the \BMS algorithm.

\item In the last family, 
  the leading monomials of the ideal of relations
  are all the monomials of degree $d$. Thus, its staircase
  is a simplex and has size $\binom{d+1}{2}=\frac{d\,(d+1)}{2}$ in
  two variables. In three variables, the staircase has
  size $\binom{d+2}{3}=\frac{d\,(d+1)\,(d+2)}{6}$. This family
  will be called \emph{Simplex}. It should be the best case for both
  the \sFGLM and the \BMS algorithms.
\end{itemize}
For all these families,
we called the algorithms with the $\DRL(y\prec x)$ ordering.

For the \BMS algorithm, we used
Proposition~\ref{prop:upperbound} to
estimate sharply the stopping monomial. For the \sFGLM algorithm,
we took all
the monomials of the largest degree appearing in the staircase and
the minimal \gb.

\subsection{Counting the number of table queries}
Thanks to the Proposition~\ref{prop:upperbound} giving a monomial $M$
such that at
step $M$, the \BMS algorithm recovers a \gb of the ideal of relations
of the input sequence, we can deduce the
following proposition.
\begin{proposition}\label{prop:bms_queries}
  Let $\bu=(u_{\bi})_{\bi\in\N^n}$
  be a sequence and $\cG$ be a minimal \gb of its ideal of
  relations for a total degree ordering.
  
  Let $d_S$ be the greatest
  degree of the elements in the staircase of $\cG$, $d_{\cG}$ be the
  greatest degree of the elements in $\cG$ and
  $d_{\max}=\max(d_S,d_{\cG})$.

  Let $\cS(d)$ be the
  simplex of all monomials of degree $d$.

  Then, the \BMS algorithm needs to perform at least
  $\#\,\cS(d_S+d_{\max}-1)=\binom{n+d_S+d_{\max}-1}{n}$ and at most
  $\#\,\cS(d_S+d_{\max})=\binom{n+d_S+d_{\max}}{n}$ queries to
  $\bu$.

  The \sFGLM algorithm called on $T=\cS(d_{\max})$
  the set of all of monomials of
  degree at most $d_{\max}$ needs to perform
  $\#\,\cS(2\,d_{\max})=\binom{n+2\,d_{\max}}{n}$ queries to $\bu$.

  For $n$ fixed, these numbers grow as $O(d_{\max}^n)$.
\end{proposition}

\begin{figure}[htbp!]
  \pgfplotsset{
    small,
    width=12cm,
    height=7.5cm,
    legend cell align=left,
    legend columns=4,
    legend style={at={(-0.05,0.98)},anchor=south
      west,font=\scriptsize,
    }
  }
  \centering
  \begin{tikzpicture}[baseline]
    \begin{axis}[
      ymode=log, xlabel={$d$}, xlabel
      style={at={(0.95,0.1)}},
      xmin=3.8,xmax=20.2,ymin=2.8,ymax=35,
      xtick={2,...,20},
      ytick={1,2,3,4,5,6,7,8,9,10,20,30,40,50,60,70,80,90,
        100,200,300,400,500,600,700,800,900,1000},
      yticklabels={},
      extra y ticks={1,5,10,50,100,500},
      extra y tick labels={1,5,10,50,100,500},
      ylabel={\#\,Queries/\#\,S},
      ylabel style={at={(0.08,0.750)}}
      ]
      \addlegendimage{empty legend}
      \addlegendentry{}
      \addlegendimage{legend image with text=Rectangle}
      \addlegendentry{}
      \addlegendimage{legend image with text=\textsc{L} shape}
      \addlegendentry{}
      \addlegendimage{legend image with text=Simplex}
      \addlegendentry{}
      
      \addlegendimage{empty legend}
      \addlegendentry{\sFGLM}
      \addlegendimage{empty legend}
      \addlegendentry{}
      \addplot[thick,every mark/.append style={solid},
      mark=triangle*,dashed,blue]
      plot coordinates {
        (2,15/3) (3,28/5) (4,45/7) (5,66/9) 
        (6,91/11) (7,120/13) (8,153/15) (9,190/17) 
        (10,231/19) (11,276/21) (12,325/23) (13,378/25) 
        (14,435/27) (15,496/29) (16,561/31) (17,630/33) 
        (18,703/35) (19,780/37) (20,861/39) 
      };
      \addlegendentry{} 
      \addplot[thick,every mark/.append style={solid},
      mark=triangle*,dashed,green]
      plot coordinates {
        (2,15/3) (3,28/6) (4,45/10) (5,66/15) 
        (6,91/21) (7,120/28) (8,153/36) (9,190/45) 
        (10,231/55) (11,276/66) (12,325/78) (13,378/91) 
        (14,435/105) (15,496/120) (16,561/136) (17,630/153) 
        (18,703/171) (19,780/190) (20,861/210) 
      };
      \addlegendentry{} 

      \addlegendimage{empty legend}
      \addlegendentry{\BMS}

      \addlegendimage{empty legend}
      \addlegendentry{}
      \addplot[thick,every mark/.append style={solid,rotate=180},
      mark=triangle*,dotted,blue]
      plot coordinates {
        (2,10/3) (3,21/5) (4,36/7) (5,55/9) 
        (6,78/11) (7,105/13) (8,136/15) (9,171/17) 
        (10,210/19) (11,253/21) (12,300/23) (13,351/25) 
        (14,406/27) (15,465/29) (16,528/31) (17,595/33) 
        (18,666/35) (19,741/37) (20,820/39) 
      };
      \addlegendentry{}
      \addplot[thick,every mark/.append style={solid,rotate=180},
      mark=triangle*,dotted,green]
      plot coordinates {
        (2,10/3) (3,21/6) (4,36/10) (5,55/15) 
        (6,78/21) (7,105/28) (8,136/36) (9,171/45) 
        (10,210/55) (11,253/66) (12,300/78) (13,351/91) 
        (14,406/105) (15,465/120) (16,528/136) (17,595/153) 
        (18,666/171) (19,741/190) (20,820/210) 
      };
      \addlegendentry{}

      \addlegendimage{empty legend}
      \addlegendentry{Both algorithms}
      \addplot[thick,every mark/.append style={solid},
      mark=square*,red]
      plot coordinates {
        (4,45/8) (5,66/10) 
        (6,120/18) (7,153/21) (8,231/32) (9,276/36) 
        (10,378/50) (11,435/55) (12,561/72) (13,630/78) 
        (14,780/98) (15,861/105) (16,1035/128) (17,1128/136) 
        (18,1326/162) (19,1431/171) (20,1653/200) 
      };
      \addlegendentry{}
    \end{axis}
  \end{tikzpicture}
  \caption{Number of table queries (\textsc{2D})}
  \label{fig:queries2D}
\end{figure}
\begin{figure}[htbp!]
  \pgfplotsset{
    small,
    width=12cm,
    height=7.5cm,
    legend cell align=left,
    legend columns=4,
    legend style={at={(-0.05,0.98)},anchor=south
      west,font=\scriptsize,
    }
  }
  \centering
  \begin{tikzpicture}[baseline]
    \begin{axis}[
      ymode=log, xlabel={$d$}, xlabel
      style={at={(0.95,0.1)}},
      xmin=3.7,xmax=10.2,ymin=3.8,ymax=65,
      xtick={2,...,10},
      ytick={1,2,3,4,5,6,7,8,9,10,20,30,40,50,60,70,80,90,
        100,200,300,400,500,600,700,800,900,
        1000,2000,3000,4000,5000,6000,7000,8000,9000,
        10000,20000},
      yticklabels={},
      extra y ticks={1,5,10,50,100,500,1000,5000,10000},
      extra y tick labels={1,5,10,50,100,500,1000,5000,10000},
      ylabel={\#\,Queries/\#\,S},
      ylabel style={at={(0.08,0.75)}}
      ]
      \addlegendimage{empty legend}
      \addlegendentry{}
      \addlegendimage{legend image with text=Rectangle}
      \addlegendentry{}
      \addlegendimage{legend image with text=\textsc{L} shape}
      \addlegendentry{}
      \addlegendimage{legend image with text=Simplex}
      \addlegendentry{}
      
      \addlegendimage{empty legend}
      \addlegendentry{\sFGLM}
      \addlegendimage{empty legend}
      \addlegendentry{}
      \addplot[thick,every mark/.append style={solid},
      mark=triangle*,dashed,blue]
      plot coordinates {
        (2,35/4) (3,84/7) (4,165/10) (5,286/13) 
        (6,455/16) (7,680/19) (8,969/22) (9,1330/25) 
        (10,1771/28) 
      };
      \addlegendentry{}
      \addplot[thick,every mark/.append style={solid},
      mark=triangle*,dashed,green]
      plot coordinates {
        (2,35/4) (3,84/10) (4,165/20) (5,286/35) 
        (6,455/56) (7,680/84) (8,969/120) (9,1330/165) 
        (10,1771/220) 
      };
      \addlegendentry{}
      
      \addlegendimage{empty legend}
      \addlegendentry{\BMS}
      \addlegendimage{empty legend}
      \addlegendentry{}
      \addplot[thick,every mark/.append style={solid,rotate=180},
      mark=triangle*,dotted,blue]
      plot coordinates {
        (2,20/4) (3,56/7) (4,120/10) (5,220/13) 
        (6,364/16) (7,560/19) (8,816/22) (9,1140/25) 
        (10,1540/28) 
      };
      \addlegendentry{}
      \addplot[thick,every mark/.append style={solid,rotate=180},
      mark=triangle*,dotted,green]
      plot coordinates {
        (2,20/4) (3,56/10) (4,120/20) (5,220/35) 
        (6,364/56) (7,560/84) (8,816/120) (9,1140/165) 
        (10,1540/220) 
      };
      \addlegendentry{}

      \addlegendimage{empty legend}
      \addlegendentry{Both algorithms}
      \addplot[thick,every mark/.append style={solid},
      mark=square*,red]
      plot coordinates {
        (4,286/16) (5,455/20) 
        (6,969/36) (7,1771/63) (8,2925/96) (9,3654/108) 
        (10,6545/200) 
      };
      \addlegendentry{}
    \end{axis}
  \end{tikzpicture}
  \caption{Number of table queries (\textsc{3D})}
  \label{fig:queries3D}
\end{figure}
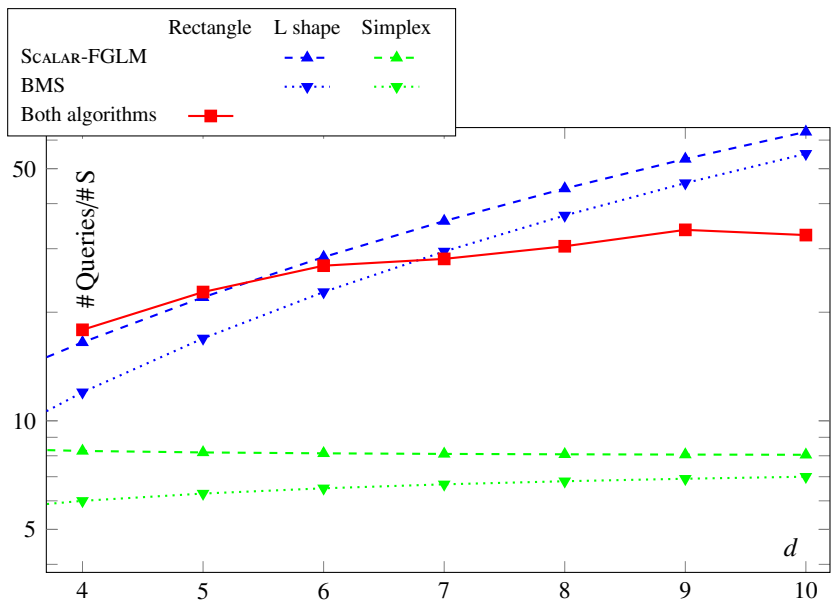

In the experiments of Figures~\ref{fig:queries2D} and~\ref{fig:queries3D},
we report on the
ratio between the numbers of queries and the size of the staircase for
the three families of polynomials.

Not surprisingly, the \sFGLM algorithm always performs the most
queries. This is due to the fact that in
Proposition~\ref{prop:bms_queries}, either $d_{\cG}=d_S+1$ or
$d_S\geq d_{\cG}$, hence $d_{\max}\in\{d_S-1,d_S\}$ and
$d_S+d_{\max}\in\{2\,d_{\max}-1,2\,d_{\max}\}$.

Though, we can see that for the Rectangle family, each algorithm
performs exactly as many queries as the other.

For the Rectangle and Simplex families, the size of the staircase and
the number of queries grow like $O(d^n)$, where $n=2,3$, the
dimension. This is why the ratio seems rather constant.

However, for
the \textsc{L} shape family, the size of the staircase only grows as
$O(d)$ while the number of queries grows as $O(d^n)$.
This is also confirmed by our
experiments, where the ratio between the number of queries and the
size of the staircase grows much faster in dimension $3$ than in
dimension $2$.

In fact, each algorithm performs as many queries for the \textsc{L}
shape family as for the Simplex family. Thus, we can see that neither
is able to take profit from the size of the staircase.

\subsection{Counting the number of basic operations}
The complexity of the \BMS algorithm has been studied
in~\cite{Sakata09}.
\begin{proposition}\label{prop:bms_basicop}
  Let $\bu=(u_{\bi})_{\bi\in\N^n}$
  be a sequence, $\cG$ be a minimal \gb of its ideal of
  relations for a total degree ordering and $S$ be the staircase of $\cG$.

  Then, the \BMS algorithm performs at most
  $O\pare{(\#\,S)^2\,\LM(\cG)}$ operations to recover the
  ideal of relations of $\bu$.  
\end{proposition}

The \sFGLM computes the column rank profile of a matrix of size
$\#\,\cS(d_{\max})$. Then, it solves as many linear systems with the submatrix
of size $\#\,S$
as there are polynomials in the \gb. All in all, we have the following
result.
\begin{proposition}\label{prop:sfglm_basicop}
  Let $\bu=(u_{\bi})_{\bi\in\N^n}$
  be a sequence, $\cG$ be a reduced \gb of its ideal of
  relations for a total degree ordering and $S$ be the staircase of $\cG$.
  Let $d_{\max}$ be the maximal degree of the elements of $S$ and $\cG$.

  Then, the number of operations performed by the \sFGLM algorithm to
  recover the ideal of relations of $\bu$ is at most
  $O\pare{(\#\,\cS(d_{\max}))^3+(\#\,S)^2\,\#\,\LM(\cG)}$. 
\end{proposition}

In the following Figures~\ref{fig:basicop2D} and~\ref{fig:basicop3D},
we report on the ratio between the number of basic operations and the
cube of the size of the staircase.

\begin{figure}[htbp!]
  \pgfplotsset{
    small,
    width=12cm,
    height=7.5cm,
    legend cell align=left,
    legend columns=4,
    legend style={at={(-0.05,0.98)},anchor=south
      west,font=\scriptsize,
    }
  }
  \centering
  \begin{tikzpicture}[baseline]
    \begin{axis}[
      ymode=log,
      xlabel={$d$},
      xlabel style={at={(0.95,0.1)}},
      xmin=3.8,xmax=20.2,
      ymin=0.4,ymax=220,
      xtick={4,...,20},
      ytick={0.1,0.2,0.3,0.4,0.5,0.6,0.7,0.8,0.9,1,2,3,4,5,6,7,8,9,
        10,20,30,40,50,60,70,80,90,100,200,300,400,500,600,700,800,900,
      1000,2000,3000,4000,5000,6000,7000,8000,9000,10000,20000},
      yticklabels={},
      extra y ticks={0.1,0.5,1,5,10,50,100,500,1000,5000,10000},
      extra y tick labels={0.1,0.5,1,5,10,50,100,500,1000,5000,10000},
      ylabel={\#\,Basic Op/\#\,S$^3$},
      ylabel style={at={(0.08,0.75)}},
      ]
      \addlegendimage{empty legend}
      \addlegendentry{}
      \addlegendimage{legend image with text=Rectangle}
      \addlegendentry{}
      \addlegendimage{legend image with text=\textsc{L} shape}
      \addlegendentry{}
      \addlegendimage{legend image with text=Simplex}
      \addlegendentry{}

      \addlegendimage{empty legend}
      \addlegendentry{\sFGLM}
      \addplot[thick,every mark/.append style={solid},
      mark=triangle*,dashed,red] plot coordinates {
        (4,1255/8^3) (5,3250/10^3) 
        (6,16170/18^3) (7,31092/21^3) (8,97697/32^3) (9,160267/36^3) 
        (10,390280/50^3) (11,580810/55^3) (12,1202772/72^3) (13,1676325/78^3) 
        (14,3103205/98^3) (15,4126045/105^3) (16,7035430/128^3)
        (17,9028762/136^3) 
        (18,14457627/162^3) (19,18048627/171^3) (20,27502685/200^3) 
      };
      \addlegendentry{}
      \addplot[thick,every mark/.append style={solid},
      mark=triangle*,dashed,blue] plot coordinates {
        (2,100/3^3) (3,420/5^3) (4,1288/7^3) (5,3329/9^3) 
        (6,7623/11^3) (7,15885/13^3) (8,30675/15^3) (9,55638/17^3) 
        (10,95774/19^3) (11,157738/21^3) (12,250170/23^3) (13,384055/25^3) 
        (14,573113/27^3) (15,834219/29^3) (16,1187853/31^3) (17,1658580/33^3) 
        (18,2275560/35^3) (19,3073088/37^3) (20,4091164/39^3) 
      };
      \addlegendentry{}
      \addplot[thick,every mark/.append style={solid},
      mark=triangle*,dashed,green] plot coordinates {
        (2,100/3^3) (3,549/6^3) (4,1965/10^3) (5,5480/15^3) 
        (6,12957/21^3) (7,27230/28^3) (8,52374/36^3) (9,94005/45^3) 
        (10,159610/55^3) (11,258907/66^3) (12,404235/78^3) (13,610974/91^3) 
        (14,897995/105^3) (15,1288140/120^3) (16,1808732/136^3)
        (17,2492115/153^3) 
        (18,3376224/171^3) (19,4505185/190^3) (20,5929945/210^3) 
      };
      \addlegendentry{}

      \addlegendimage{empty legend}
      \addlegendentry{\BMS}
      \addplot[thick,every mark/.append style={solid,rotate=180},
      mark=triangle*,dotted,red] plot coordinates {
        (4,2978/8^3) (5,5202/10^3) 
        (6,18677/18^3) (7,27720/21^3) (8,69808/32^3) (9,94593/36^3) 
        (10,198745/50^3) (11,254459/55^3) (12,473240/72^3) (13,582804/78^3) 
        (14,992697/98^3) (15,1188420/105^3) (16,1894464/128^3)
        (17,2219609/136^3) 
        (18,3359830/162^3) (19,3870335/171^3) (20,5620488/200^3) 
      };
      \addlegendentry{}
      \addplot[thick,every mark/.append style={solid,rotate=180},
      mark=triangle*,dotted,blue] plot coordinates {
        (2,467/3^3) (3,1371/5^3) (4,3008/7^3) (5,5603/9^3) 
        (6,9380/11^3) (7,14563/13^3) (8,21376/15^3) (9,30043/17^3) 
        (10,40788/19^3) (11,53835/21^3) (12,69408/23^3) (13,87731/25^3) 
        (14,109028/27^3) (15,133523/29^3) (16,161440/31^3) (17,193003/33^3) 
        (18,228436/35^3) (19,267963/37^3) (20,311808/39^3) 
      };
      \addlegendentry{}
      \addplot[thick,every mark/.append style={solid,rotate=180},
      mark=triangle*,dotted,green] plot coordinates {
        (2,427/3^3) (3,1759/6^3) (4,5241/10^3) (5,12860/15^3) 
        (6,27552/21^3) (7,53414/28^3) (8,95823/36^3) (9,161690/45^3) 
        (10,259672/55^3) (11,400330/66^3) (12,596325/78^3) (13,862630/91^3) 
        (14,1216684/105^3) (15,1678622/120^3) (16,2271453/136^3)
        (17,3021248/153^3) 
        (18,3957182/171^3) (19,5112366/190^3) (20,6522843/210^3) 
      };
      \addlegendentry{}
    \end{axis}
  \end{tikzpicture}
  \caption{Number of basic operations (\textsc{2D})}
  \label{fig:basicop2D}
\end{figure}
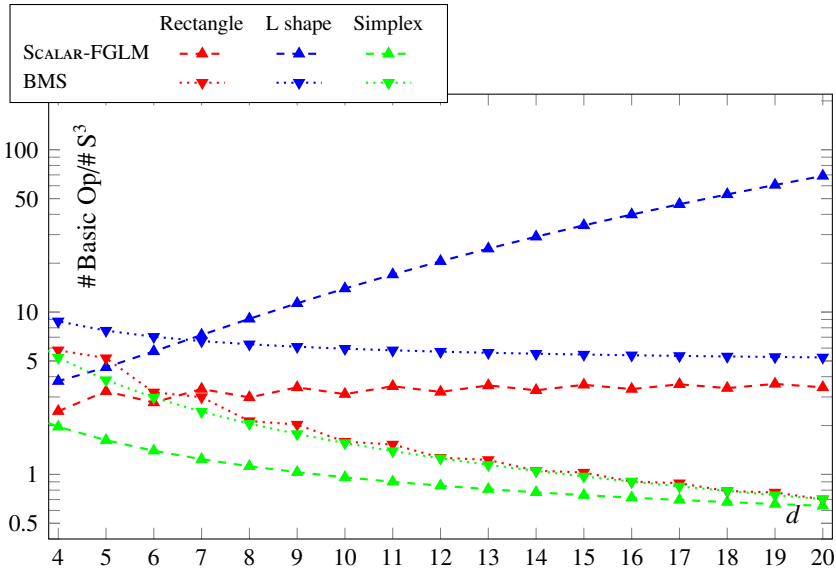

For the Rectangle family, we have $\#\,S\in O(d^n)$,
$\#\,\cS(d_{\max})\in O(d^n)$ and $\LM(\cG)=3$ so that
$(\#\,S)^2\,\#\,\LM(\cG)\in O(d^{2\,n})$. This is why, we can see,
first, a constant ratio between
the number of basic operations done by the \sFGLM algorithm
and the size of the staircase and, then, a decreasing ratio for the \BMS
algorithm. An analogous analysis explains why, for the \textsc{L}
shape family, the ratio is increasing
for the \sFGLM algorithm and quite constant for the \BMS algorithm.

Unexpectedly, the \sFGLM algorithm performs fewer basic operations
than the \BMS algorithm for the Simplex family. This is mainly due to
the fact that, for this family, the term $(\#\,S)^2\,\#\,\LM(\cG)$ is
in fact larger than $(\#\,\cS(d_{\max}))^3$.
\begin{figure}[htbp!]
  \pgfplotsset{
    small,
    width=12cm,
    height=7.5cm,
    legend cell align=left,
    legend columns=4,
    legend style={at={(-0.05,0.98)},anchor=south
      west,font=\scriptsize,
    }
  }
  \centering
  \begin{tikzpicture}[baseline]
    \begin{axis}[
      ymode=log,
      xlabel={$d$},
      xlabel style={at={(0.95,0.1)}},
      xmin=3.7,xmax=10.2,
      ymin=0.8,ymax=750,
      xtick={2,...,10},
      ytick={0.1,0.2,0.3,0.4,0.5,0.6,0.7,0.8,0.9,1,2,3,4,5,6,7,8,9,10,
        20,30,40,50,60,70,80,90,100,200,300,400,500,600,700,800,900,
        1000,2000,3000},
      yticklabels={},
      extra y ticks={0.1,0.5,1,5,10,50,100,500,1000},
      extra y tick labels={0.1,0.5,1,5,10,50,100,500,1000},
      ylabel={\#\,Basic Op/\#\,S$^3$},
      ylabel style={at={(0.08,0.75)}},
      ]
      \addlegendimage{empty legend}
      \addlegendentry{}
      \addlegendimage{legend image with text=Rectangle}
      \addlegendentry{}
      \addlegendimage{legend image with text=\textsc{L} shape}
      \addlegendentry{}
      \addlegendimage{legend image with text=Simplex}
      \addlegendentry{}

      \addlegendimage{empty legend}
      \addlegendentry{\sFGLM}
      \addplot[thick,every mark/.append style={solid},
      mark=triangle*,dashed,red] plot coordinates {
        (4,58468/16^3) (5,196394/20^3) 
        (6,1491458/36^3) (7,7780660/63^3) (8,31350363/96^3) (9,58451620/108^3) 
        (10,303054484/200^3) 
      };
      \addlegendentry{}
      \addplot[thick,every mark/.append style={solid},
      mark=triangle*,dashed,blue] plot coordinates {
        (2,369/4^3) (3,2743/7^3) (4,14255/10^3) (5,57955/13^3) 
        (6,195542/16^3) (7,570929/19^3) (8,1486628/22^3) (9,3528845/25^3) 
        (10,7761655/28^3) 
      };
      \addlegendentry{}
      \addplot[thick,every mark/.append style={solid},
      mark=triangle*,dashed,green] plot coordinates {
        (2,453/4^3) (3,4370/10^3) (4,25385/20^3) (5,107695/35^3) 
        (6,368102/56^3) (7,1073828/84^3) (8,2774390/120^3) (9,6510845/165^3) 
        (10,14131315/220^3) 
      };
      \addlegendentry{}

      \addlegendimage{empty legend}
      \addlegendentry{\BMS}
      \addplot[thick,every mark/.append style={solid,rotate=180},
      mark=triangle*,dotted,red] plot coordinates {
        (4,55942/16^3) (5,113816/20^3) 
        (6,477592/36^3) (7,1820818/63^3) (8,5275811/96^3) (9,7677599/108^3) 
        (10,32891946/200^3) 
      };
      \addlegendentry{}
      \addplot[thick,every mark/.append style={solid,rotate=180},
      mark=triangle*,dotted,blue] plot coordinates {
        (2,1988/4^3) (3,9217/7^3) (4,29636/10^3) (5,75497/13^3) 
        (6,164592/16^3) (7,321217/19^3) (8,577132/22^3) (9,972521/25^3) 
        (10,1556952/28^3) 
      };
      \addlegendentry{}
      \addplot[thick,every mark/.append style={solid,rotate=180},
      mark=triangle*,dotted,green] plot coordinates {
        (2,1774/4^3) (3,14591/10^3) (4,75757/20^3) (5,298521/35^3) 
        (6,964815/56^3) (7,2689885/84^3) (8,6679544/120^3) (9,15125328/165^3) 
        (10,31763926/220^3) 
      };
      \addlegendentry{}
    \end{axis}
  \end{tikzpicture}
  \caption{Number of basic operations (\textsc{3D})}
  \label{fig:basicop3D}
\end{figure}
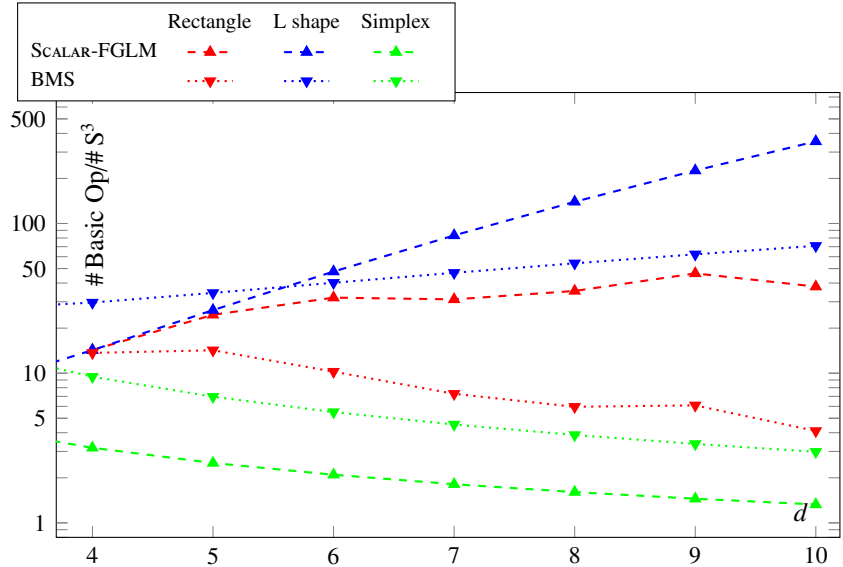

We now compare the ratio
between the number of
basic operations and the number of queries made by each algorithm in
Figures~\ref{fig:basicop/queries2D} and~\ref{fig:basicop/queries3D}.
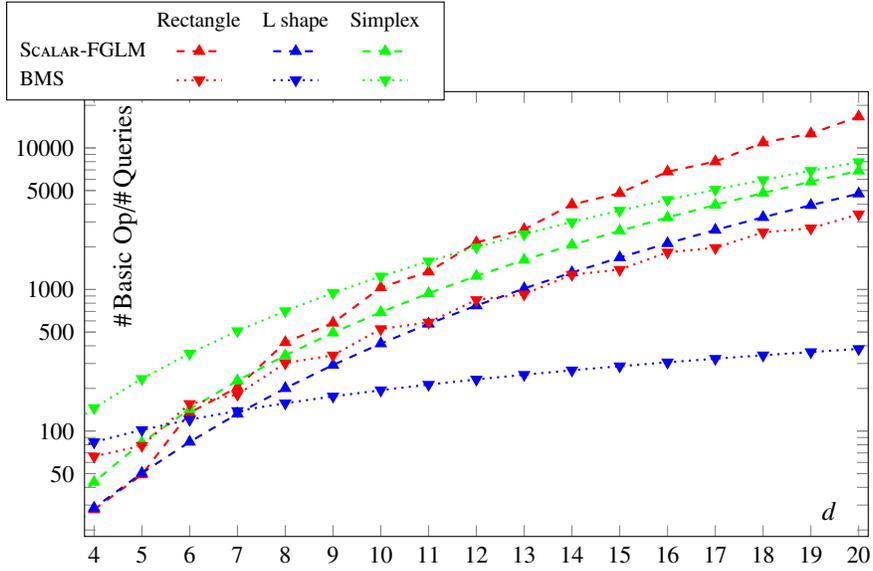
\begin{figure}[htbp!]
  \pgfplotsset{
    small,
    width=12cm,
    height=7.5cm,
    legend cell align=left,
    legend columns=4,
    legend style={at={(-0.1,0.98)},anchor=south west,font=\scriptsize,
    }
  }
  \centering
  \begin{tikzpicture}[baseline]
    \begin{axis}[
      ymode=log,
      xlabel={$d$},
      xlabel style={at={(0.95,0.1)}},
      xmin=3.8,xmax=20.2,
      ymin=18,ymax=25000,
      xtick={3,...,20},
      ytick={1,2,3,4,5,6,7,8,9,10,20,30,40,50,60,70,80,90,100,
        200,300,400,500,600,700,800,900,1000,
        2000,3000,4000,5000,6000,7000,8000,9000,10000,
        20000,30000,40000,50000,60000,70000,80000,90000,100000,
        200000},
      yticklabels={},
      extra y ticks={5,10,50,100,500,1000,5000,10000,50000,100000},
      extra y tick labels={5,10,50,100,500,1000,5000,10000,50000,100000},
      ylabel={\#\,Basic Op/\#\,Queries},
      ylabel style={at={(0.08,0.72)}},
      ]
      \addlegendimage{empty legend}
      \addlegendentry{}
      \addlegendimage{legend image with text=Rectangle}
      \addlegendentry{}
      \addlegendimage{legend image with text=\textsc{L} shape}
      \addlegendentry{}
      \addlegendimage{legend image with text=Simplex}
      \addlegendentry{}

      \addlegendimage{empty legend}
      \addlegendentry{\sFGLM}
      \addplot[thick,every mark/.append style={solid},
      mark=triangle*,dashed,red] plot coordinates {
        (4,1255/45) (5,3250/66) 
        (6,16170/120) (7,31092/153) (8,97697/231) (9,160267/276) 
        (10,390280/378) (11,580810/435) (12,1202772/561) (13,1676325/630) 
        (14,3103205/780) (15,4126045/861) (16,7035430/1035) (17,9028762/1128) 
        (18,14457627/1326) (19,18048627/1431) (20,27502685/1653) 
      };
      \addlegendentry{}
      \addplot[thick,every mark/.append style={solid},
      mark=triangle*,dashed,blue] plot coordinates {
        (2,100/15) (3,420/28) (4,1288/45) (5,3329/66) 
        (6,7623/91) (7,15885/120) (8,30675/153) (9,55638/190) 
        (10,95774/231) (11,157738/276) (12,250170/325) (13,384055/378) 
        (14,573113/435) (15,834219/496) (16,1187853/561) (17,1658580/630) 
        (18,2275560/703) (19,3073088/780) (20,4091164/861) 
      };
      \addlegendentry{}
      \addplot[thick,every mark/.append style={solid},
      mark=triangle*,dashed,green] plot coordinates {
        (2,100/15) (3,549/28) (4,1965/45) (5,5480/66) 
        (6,12957/91) (7,27230/120) (8,52374/153) (9,94005/190) 
        (10,159610/231) (11,258907/276) (12,404235/325) (13,610974/378) 
        (14,897995/435) (15,1288140/496) (16,1808732/561) (17,2492115/630) 
        (18,3376224/703) (19,4505185/780) (20,5929945/861) 
      };
      \addlegendentry{}

      \addlegendimage{empty legend}
      \addlegendentry{\BMS}
      \addplot[thick,every mark/.append style={solid,rotate=180},
      mark=triangle*,dotted,red] plot coordinates {
        (4,2978/45) (5,5202/66) 
        (6,18677/120) (7,27720/153) (8,69808/231) (9,94593/276) 
        (10,198745/378) (11,254459/435) (12,473240/561) (13,582804/630) 
        (14,992697/780) (15,1188420/861) (16,1894464/1035) (17,2219609/1128) 
        (18,3359830/1326) (19,3870335/1431) (20,5620488/1653) 
      };
      \addlegendentry{}
      \addplot[thick,every mark/.append style={solid,rotate=180},
      mark=triangle*,dotted,blue] plot coordinates {
        (2,467/10) (3,1371/21) (4,3008/36) (5,5603/55) 
        (6,9380/78) (7,14563/105) (8,21376/136) (9,30043/171) 
        (10,40788/210) (11,53835/253) (12,69408/300) (13,87731/351) 
        (14,109028/406) (15,133523/465) (16,161440/528) (17,193003/595) 
        (18,228436/666) (19,267963/741) (20,311808/820) 
      };
      \addlegendentry{}
      \addplot[thick,every mark/.append style={solid,rotate=180},
      mark=triangle*,dotted,green] plot coordinates {
        (2,427/10) (3,1759/21) (4,5241/36) (5,12860/55) 
        (6,27552/78) (7,53414/105) (8,95823/136) (9,161690/171) 
        (10,259672/210) (11,400330/253) (12,596325/300) (13,862630/351) 
        (14,1216684/406) (15,1678622/465) (16,2271453/528) (17,3021248/595) 
        (18,3957182/666) (19,5112366/741) (20,6522843/820) 
      };
      \addlegendentry{}
    \end{axis}
  \end{tikzpicture}
  \caption{Number of basic operations by queries (\textsc{2D})}
  \label{fig:basicop/queries2D}
\end{figure}
\begin{figure}[htbp!]
  \pgfplotsset{
    small,
    width=12cm,
    height=7.5cm,
    legend cell align=left,
    legend columns=4,
    legend style={at={(-0.1,0.98)},anchor=south west,font=\scriptsize,
    }
  }
  \centering
  \begin{tikzpicture}[baseline]
    \begin{axis}[
      ymode=log,
      xlabel={$d$},
      xlabel style={at={(0.95,0.1)}},
      xmin=3.7,xmax=10.2,
      ymin=55,ymax=65000,
      xtick={2,...,12},
      ytick={1,2,3,4,5,6,7,8,9,10,20,30,40,50,60,70,80,90,100,
        200,300,400,500,600,700,800,900,1000,2000,3000,4000,5000,6000,
        7000,8000,9000,10000,20000,30000,40000,50000,60000},
      yticklabels={},
      extra y ticks={5,10,50,100,500,1000,5000,10000,50000},
      extra y tick labels={5,10,50,100,500,1000,5000,10000,50000},
      ylabel={\#\,Basic Op/\#\,Queries},
      ylabel style={at={(0.08,0.72)}},
      ]
      \addlegendimage{empty legend}
      \addlegendentry{}
      \addlegendimage{legend image with text=Rectangle}
      \addlegendentry{}
      \addlegendimage{legend image with text=\textsc{L} shape}
      \addlegendentry{}
      \addlegendimage{legend image with text=Simplex}
      \addlegendentry{}

      \addlegendimage{empty legend}
      \addlegendentry{\sFGLM}
      \addplot[thick,every mark/.append style={solid},
      mark=triangle*,dashed,red] plot coordinates {
        (4,58468/286) (5,196394/455) 
        (6,1491458/969) (7,7780660/1771) (8,31350363/2925) (9,58451620/3654) 
        (10,303054484/6545) 
      };
      \addlegendentry{}
      \addplot[thick,every mark/.append style={solid},
      mark=triangle*,dashed,blue] plot coordinates {
        (2,369/35) (3,2743/84) (4,14255/165) (5,57955/286) 
        (6,195542/455) (7,570929/680) (8,1486628/969) (9,3528845/1330) 
        (10,7761655/1771) 
      };
      \addlegendentry{}
      \addplot[thick,every mark/.append style={solid},
      mark=triangle*,dashed,green] plot coordinates {
        (2,453/35) (3,4370/84) (4,25385/165) (5,107695/286) 
        (6,368102/455) (7,1073828/680) (8,2774390/969) (9,6510845/1330) 
        (10,14131315/1771) 
      };
      \addlegendentry{}

      \addlegendimage{empty legend}
      \addlegendentry{\BMS}
      \addplot[thick,every mark/.append style={solid,rotate=180},
      mark=triangle*,dotted,red] plot coordinates {
        (4,55942/286) (5,113816/455) 
        (6,477592/969) (7,1820818/1771) (8,5275811/2925) (9,7677599/3654) 
        (10,32891946/6545) 
      };
      \addlegendentry{}
      \addplot[thick,every mark/.append style={solid,rotate=180},
      mark=triangle*,dotted,blue] plot coordinates {
        (2,1988/20) (3,9217/56) (4,29636/120) (5,75497/220) 
        (6,164592/364) (7,321217/560) (8,577132/816) (9,972521/1140) 
        (10,1556952/1540) 
      };
      \addlegendentry{}
      \addplot[thick,every mark/.append style={solid,rotate=180},
      mark=triangle*,dotted,green] plot coordinates {
        (2,1774/20) (3,14591/56) (4,75757/120) (5,298521/220) 
        (6,964815/364) (7,2689885/560) (8,6679544/816) (9,15125328/1140) 
        (10,31763926/1540) 
      };
      \addlegendentry{}
    \end{axis}
  \end{tikzpicture}
  \caption{Number of basic operations by queries (\textsc{3D})}
  \label{fig:basicop/queries3D}
\end{figure}
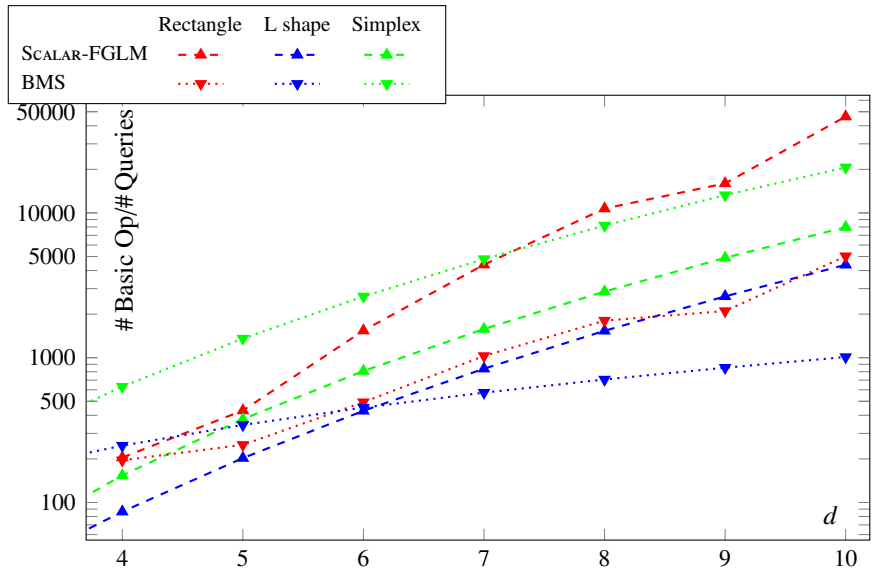

As we can see, beside for the Simplex family where the \sFGLM
performed fewer operations but more queries than the \BMS algorithm,
the polynomial arithmetic of the \BMS
algorithm allows it to have a much better behavior than the \sFGLM algorithm.

This reinforces the conviction that an hybrid approach between the
\BMS and the \sFGLM algorithm or a fast multi-Hankel solver
should be investigated.



\bibliographystyle{elsarticle-harv} 
\addcontentsline{toc}{section}{References}
\bibliography{biblio}






\end{document}